\documentclass[12pt]{article}

\usepackage[margin=2.5cm]{geometry}

\usepackage{amsmath,amsthm,amsfonts,amscd,amssymb,bbm,mathrsfs,enumerate,url}

\usepackage{graphicx,tikz}
\usepackage[pdfborder={0 0 0}]{hyperref}
\usetikzlibrary{matrix,arrows}
\usetikzlibrary{decorations.pathreplacing}
\usetikzlibrary{decorations.pathmorphing}
\usetikzlibrary{patterns,fadings}
\usepackage{braket,thmtools}

\usepackage[all]{xy}

\newcounter{axiomhk}
\newcounter{axiomss}
\def\ol{\overline}

\def\RR{{\mathbb R}}
\def\CC{{\mathbb C}}
\def\NN{{\mathbb N}}

\def\A{{\mathcal A}}
\def\B{{\mathcal B}}
\def\C{{\mathcal C}}
\def\D{{\mathcal D}}

\def\H{{\mathcal H}}

\def\K{{\mathcal K}}

\def\P{{\mathcal P}}

\def\S{{\mathcal S}}

\def\a{\alpha}
\def\b{\beta}

\def\e{\varepsilon}

\def\k{\kappa}

\def\l{\lambda}

\def\hm{\mathfrak{H}_m}

\def\Ad{{\hbox{\rm Ad\,}}}

\def\1{{\mathbbm 1}}

\def\Span{{\mathrm{Span}\,}}

\def\uone{{\rm U(1)}}

\def\diff{{\rm Diff}}

\def\diffs1{\diff(S^1)}

\def\supp{{\rm supp\,}}

\def\slim{{{\mathrm{s}\textrm{-}\lim}\,}}

\def\psl2r{{\rm PSL}(2,\RR)}
\def\sl2r{{\rm SL}(2,\RR)}
\def\su11{{\rm SU}(1,1)}
\def\2dmob{{\overline{\psl2r}\times\overline{\psl2r}}}
\def\<{\langle}
\def\>{\rangle}

\def\Im{\mathrm{Im}\,}

\def\im{\mathrm{Im}\,}

\def\arcsinh{\mathrm{arcsinh}}

\def\poincare{{\P^\uparrow_+}}

\def\tp{\pmb{p}}
\def\tP{\textbf{P}}
\def\tx{\pmb{x}}

\newcommand{\norm}[1]{\left\lVert#1\right\rVert}

\newtheorem{theorem}{Theorem}[section]

\newtheorem{corollary}[theorem]{Corollary}
\newtheorem{proposition}[theorem]{Proposition}
\newtheorem{lemma}[theorem]{Lemma}
\theoremstyle{remark}
\newtheorem{remark}[theorem]{Remark}

\title{Modular operator for null plane algebras in free fields}
\date{} 
\author{{Vincenzo Morinelli}${}^{1,}$\footnote{{\tt morinelli@math.fau.de}, ${}^\bigtriangleup$ {\tt hoyt@mat.uniroma2.it,} ${}^\bullet$ {\tt wegener@mat.uniroma2.it} }, { Yoh Tanimoto}${}^{2,\bigtriangleup}$, { Benedikt Wegener}${}^{2,\bullet}$.
}
\date{\small{
${}^{1}$Department Mathematik, FAU Erlangen-N\"urnberg, \\Cauerstra\ss e 11
91058 Erlangen, Germany\\
${}^{2}$ Dipartimento di Matematica, Universit\`a di Roma Tor Vergata
 \\  Via della Ricerca Scientifica 1, I-00133 Roma, Italy}\\ \today}

\begin{document}
\maketitle

\begin{abstract}
 We consider the algebras generated by observables in quantum field theory localized in regions in the null plane.
 For a scalar free field theory, we show that the one-particle structure can be decomposed into a continuous direct integral
 of lightlike fibres, and the modular operator decomposes accordingly.
 This implies that a certain form of QNEC is valid in free fields involving the causal completions
 of half-spaces on the null plane (null cuts).
 We also compute the relative entropy of null cut algebras with respect to the vacuum and some coherent states.

\end{abstract}

\tableofcontents

\section{Introduction}\label{intro}
The modular Hamiltonian, or the (logarithm of the half of the) modular operator of local regions in quantum field theory (QFT),
has been a focus of attention in recent years (see e.g.~\cite{ CF18,LongoLocalised,CTT17}).
On one hand, quantum-information aspects, such as the Bekenstein bound, generalized second law of thermodynamics and various null energy conditions,
are a rare guidepost in the search of quantum gravity \cite{Casini08, Wall12, FLPW16}.
On the other hand, the modular theory of von Neumann algebras allows one to define relative entropy
in QFT in a mathematically precise way \cite{OP04}, and various modular objects in QFT have been
computed in concrete examples \cite{LX18, Hollands19, CLR20}.
In particular, the modular operator of certain regions in the null plane has played
an important role in relation with the quantum null energy condition (QNEC) and
the averaged null energy condition (ANEC) \cite{CTT17, KLLS18, CF18}.
In these works, physicists consider a null cut, a region on the null plane
defined by a spacelike curve $C$ and have written a formula for the modular operator for the algebra
of a null cut (see e.g.\! \cite[(1.5)]{CTT17}):
\begin{align}\label{eq:set}
 \hat H_C = 2\pi \int d^{D-2}\pmb{x}^\perp \int_{-\infty}^\infty d\l (\l - C(\pmb{x}^\perp))T_{++}(\l,\pmb{x}^\perp),
\end{align}
where $T_{++}$ is the lightlike component of the stress-energy tensor.
This suggests that the inclusion of null cut algebras is a half-sided modular inclusion (HSMI) \cite{Wiesbrock93-1}.
Based on the latter assumption, a limited version of QNEC has been proved in \cite{CF18}.
Therefore, it is crucial to study modular object on the null plane.

Actually, the above formula must be interpreted with care:
while it seems reasonable to assume that the stress-energy tensor is an operator-valued distribution
(or even a Wightman field), it is unclear whether it can be restricted to a null plane (cf.\! \cite{Verch00, FR03}).
Furthermore, it is integrated against the unbounded function $\l - C(\pmb{x}^\perp)$,
that could be even more problematic.
For these reasons, \eqref{eq:set} cannot be considered directly as an expression for an operator on
a Hilbert space. One of the goals of this paper is to partially justify \eqref{eq:set}
using the modular theory of von Neumann algebras in the case of the free fields.

We observe that the (scalar) free field can be restricted to the null plane, with a slight
restriction to the test function. This has been known for a long time, and general properties of
observables on the null plane have been studied \cite{SS72, Driessler76-1, Driessler76-2, GLRV01, Ullrich04}.
By the Bisognano-Wichmann property \cite{BW76} and the Takesaki theorem \cite[Theorem IX.4.2]{TakesakiIII},
it is immediate that, if there are enough observables on the null plane, they split
into a (continuous) tensor product along the transverse direction.
This allows us to consider the observables on each fibre on the null plane.
These observables form a simplified quantum field theory on each lightlike fibre, and we can consider
the modular objects there.
We will show that the one-particle subspace $\hm$ of the free field with mass $m$
disintegrates as follows
\begin{align}\label{eq:ddd}
  \hm=\int_{\RR^{D-1}}^{\oplus_\RR} \H_{U(1)}d\tx_\perp.
\end{align}
where $\H_{U(1)}$ is the one-particle space of the $U(1)$-current sitting on the light ray $\{(t,t,\tx_\perp)\in\RR^{1+D}:t\in\RR\}$.
One can deduce that the modular operator of the region on the null plane is decomposed into the fibres, and its logarithm is written as a direct integral over fibres
(see \eqref{null-cut-modular-spatial-decom}):
\begin{align*}
  \log(\Delta_{H(N_C)})\simeq \int^\oplus_{\RR^{D-1}}\left( \log(\Delta_{H_{U(1)}}(\RR_+))+ 2\pi C(\pmb{x}_\perp) P_{\pmb{x}_\perp}\right)\,d\pmb{x}_\perp,
\end{align*}
This is a clear analogue of \eqref{eq:set}, where the logarithm of the modular operator is expressed as the integral
in the transverse direction of the stress-energy tensor smeared by the function $\l - C(\pmb{x}^\perp)$,
where the latter is a formal expression for the shifted dilation operator in two-dimensional conformal field theory,
which should coincide with the modular operator on each fibre.
Our formula, while we avoid taking about the stress-energy tensor, makes explicit the idea that
the modular operator decomposes into fibres. While the validity of \eqref{eq:set} is believed more generally, we point out that in general
there are not many observables that can be restricted to the null plane.
We clarify the situation from interacting models in $(1+1)$-dimensions.
Our formula allows a covariant action of such distorted dilations the null plane and, for null cuts with continuous boundary $C$, on distorted wedge region $W_C=N_C''$.

In the course of the proof (Proposition \ref{pr:hsmi}), we show that inclusions of the null cut regions are HSMI.
By \cite{CF18}, this completes the limited version of QNEC as in \eqref{def:QNEC3} in the case of the free scalar field (note that the result of \cite{CF18} is based on the assumption that these inclusions are HSMI). 

Entropy inequalities can be used to investigate features of quantum systems. For instance on the physical ground the strong subadditive property of the entropy together with the Lorentz covariance leads to a $c$-theorem for the entanglement entropy in 1+1 dimensions
and connections with the $a$-theorem are claimed in \cite{CTT17'}.
Due to the direct integral disintegration of the one-particle space \eqref{eq:ddd}, it is possible to generalize the Buchholz-Mach-Todorov endomorphism $\beta_k$ \cite{BMT88}
to the direct integral of the $U(1)$-current with $k\in C_0^\infty(X^0_-)$. Using the formula contained in \cite{LongoLocalised}, we are able to compute the relative entropy with respect the $\omega\circ \beta_k$ and $\omega$ and deduce the QNEC -- in this case to be intended $S''(t)>0$ where $S$ is the  relative entropy related of the algebra $\A(N_{C+tA})$ with respect to $\omega\circ \beta_k$ and $\omega$ and the derivative is with respect to $t$.
We also have a saturation of the strong superadditivity condition of the relative entropy considered.

This paper is organized as follows.
In Section \ref{preliminaries} we collect the basic notions such as one-particle space in terms of
standard subspaces and its second quantization.
In Section \ref{free} we discuss observables on the null plane and their transversal decomposition.
In Section \ref{section-constant-null-cut} we obtain the decomposition of the modular operator of null plane regions and prove that inclusions of null plane regions are HSMI. In Section \ref{entropy}, after recalling the notions concerning relative entropy, ANEC, QNEC and some background also from physics,
we study the relative entropy and its relation with the energy inequalities and the saturation of the strong superadditivity condition
of the null cut algebras between the BMT type states.
In Section \ref{concluding} we present concluding remarks, including the 1+1 dimensional case.

\section{Preliminaries}\label{preliminaries}
 
 In this Section we will recall the operator-algebraic formulations of the free field.
 A free field is constructed from its one-particle structure. A quantum and relativistic particle on Minkowski spacetime is a unitary positive energy representation of the Poincar\'e group. The localization property of the one-particle states is formulated in terms of the standard subspaces,
 and it translates to the localization property of the associated free fields through the second quantization.
 We will further comment on the $U(1)$-current model, which will be a convenient tool to describe the restriction of the free theory  on the null plane.

\subsection{Abstract one-particle structure}\label{sec:ss}
A real linear, closed subspace $H$ of a complex Hilbert space $\H$ is called \textbf{cyclic} if $H+iH$ is dense in $\H$
and \textbf{separating} if $H\cap i H=\{0\}$. A \textbf{standard subspace} is a real linear, closed subspace that is both cyclic and separating.
We recall below some useful properties of standard subspaces, see \cite{Longo08} for details.

It is possible to consider an analogue of the Tomita-Takesaki modular theory for standard subspaces.
For a standard subspace $H$, the Tomita operator $S_H$ is defined to be the closed anti-linear involution with dense domain $H+iH$
acting in the following way:
\begin{align*}
S_H: H+iH &\rightarrow H+iH \\
\xi+i\eta &\mapsto \xi -i \eta.
\end{align*}
The polar decomposition
\[
S_H= \Delta_H^\frac{1}{2}J_H
\]
defines the modular operator $\Delta_H$ and the modular conjugation $J_H$, and they satisfy the following relations:
\[
J_H\Delta_H J_H= \Delta_H^{-1},\ \ \Delta_H^{it}H=H \quad\text{for } t\in\RR, \ \ J_H H=H',
\]
where $H'$ is the symplectic complement of $H$:
\[
H':=\{\xi \in \H: \im\braket{\xi,\eta}=0\ \text{ for } \eta\in H\}.
\]
The symplectic complement $H'$ is a standard subspace if and only if so is $H$,
and a standard subspace $H$ satisfies $H=H''$.
The Tomita operator of the symplectic complement $H'$ is given by
\[
S_{H'}=S_H^* = \Delta_H^\frac12 J_H = J_H \Delta_H^{-\frac12}.
\]
The assignment $H\mapsto S_H$ of a closed, anti-linear, densely defined involution
is one-to-one: For such an operator $S$, the real closed subspace
$\ker (1-S)$ is a standard subspace.

One can easily deduce the covariance of standard subspaces (see \cite[Lemma 2.2]{Morinelli18}):
\begin{lemma}\label{lem:cov}
Let $H\subset\H$ be a standard subspace and $U$ be a unitary or anti-unitary operator on $\H$. Then $UH$ is standard and $U\Delta_HU^*=\Delta_{UH}^{\epsilon(U)}$ and $UJ_HU^*=J_{UH}$ where $\epsilon(U)=1$ if U is unitary and $\epsilon(U)=-1$ if $U$ is anti-unitary.
\end{lemma}

The following is an analogue of Borchers theorem \cite{Borchers92, Florig98} for standard subspaces,
see \cite[Theorem 3.15]{Longo08}.
\begin{theorem}\label{th:borchersss}
Let $H$ be a standard subspace of a Hilbert space $\H$ and $T$ a one-parameter group with positive generator
such that $T(s)H\subset H, \ s\geq 0$, then the following hold:
\begin{align*}
\Delta_H^{it}T(s)\Delta_H^{-it}&=T(e^{-2\pi t}s)\\
J_H T(s) J_H &= T(-s).
\end{align*}
\end{theorem}

We say that a pair of standard subspaces $K\subset H$ is a \textbf{half-sided modular inclusion (HSMI)} if
\begin{align*}
\Delta_{H}^{-it}K\subset K \ \text{ for } t\geq 0.
\end{align*}
Let $\mathbf{P}$ be the translation-dilation group, that is, the group of affine transformations of $\RR$,
where dilations act by
$\mathfrak{d}(2\pi t) x=e^{2\pi t}x, x \in \RR$
and translations act by
$\mathfrak{t}(s) x = x + s$. The group $\mathbf{P}$ contains also dilations centered at $1$:
$\mathfrak{d}_1(2\pi t)x=e^{2\pi t}(x-1)+1$.
A HSMI of von Neumann algebras implies the existence of a one-parameter group of unitaries with certain properties \cite{Wiesbrock93-1, AZ05}.
The following is its standard subspace version and the first part can be found in \cite[Theorem 3.21]{Longo08}.
\begin{theorem}\label{theorem-1-standard}
Let $K\subset H$ be a half-sided modular inclusion of standard subspaces of the Hilbert space $\H$, then there exists a positive energy representation of the translation-dilation group $\mathbf{P}$ given by
\[
 U(\mathfrak{d}(2\pi t))=\Delta_H^{-it},\qquad U(\mathfrak{d}_1(2\pi t))=\Delta_K^{-it} 
\]
In particular, the translations are given by
$U(\mathfrak{t}(e^{2\pi t}-1))=\Delta_H^{-it}\Delta_K^{it}$, satisfy $U(\mathfrak{t}(s))H\subset H$ for $s\geq 0$,  $U(\mathfrak{t}(1))H=K$ and
have a positive generator.

Furthermore, the generator $P$ of the translation group is $\frac{1}{2\pi}\left(\log(\Delta_K)- \log(\Delta_H) \right)$.
In general, we have the relation $\log(\Delta_{U(\mathfrak{t}(s))H}) = \Ad U(\mathfrak{t}(s))(\log(\Delta_H)) = \log(\Delta_H) + 2\pi sP$.
\end{theorem}

\begin{proof}
	We prove the last statement. The operator $\log(\Delta_K)- \log(\Delta_H)$ is essentially self-adjoint on its natural domain
	(one can prove this by first taking the G\aa rding domain). Thus, we can apply Trotter's product formula:
		\begin{align*}
			e^{it(\log(\Delta_H)-\log(\Delta_K))}&=\slim_{n \rightarrow \infty}\left(\Delta_H^{i\frac{t}{n}}\Delta_K^{-i\frac{t}{n}}  \right)^n
			= \slim_{n\rightarrow \infty}\left(U(\mathfrak{t}(e^{-2\pi \frac{t}{n}}-1))  \right)^n \\
			&=\slim_{n\rightarrow \infty}U\left(\mathfrak{t}(n(e^{-2\pi \frac{t}{n}}-1))\right) =U(\mathfrak{t}(2\pi t)). 
		\end{align*}
		
    The last relation follows from Lemma \ref{lem:cov}.
\end{proof}

\subsection{The one-particle structure of the free scalar field }\label{sec:oneparticle}
The relativistic invariance is encoded in the Poincar\'e group $\P$ on the $(D+1)$-dimensional Minkowski space time $\RR^{D+1}$,
where $D>1$.
It is the semi-direct product of the full Lorentz group $\mathcal{L}$ and the translation group $\RR^{D+1}$:
\begin{align*}
 \P=\mathcal{L}\ltimes \RR^{D+1}.
\end{align*}
The subgroup $\poincare = \mathcal{L}^\uparrow_+\ltimes \RR^{D+1}$ of time- and space orientation-preserving transformations gives the relativistic transformations
from one inertial frame to another. The causal structure is determined by the Minkowski metric and the causal complement of a region $O\subset\RR^{1+D}$ is determined as follows:
$$O'=\{x\in\RR^{1+D}:(y-x)^2<0, y\in O\}^\circ,$$
where $\circ$ denotes the open kernel.

We restrict the consideration to the scalar representations of $\poincare$ (see e.g.~\cite{Var85}).
A \textbf{scalar representation with mass $m \ge 0$ (for $D>1$)} is defined on the Hilbert space $\hm = L^2(\Omega_m,d\Omega_m)$,
where $d\Omega_m$ is the unique (up to a constant) Lorentz-invariant measure on the mass shell $\Omega_m=\{p=(p_0,\ldots,p_D)\in\RR^{D+1}: p^2=m^2, p_0\geq0\}$: 
\begin{align*}d\Omega_m=\theta(p_0)\delta(p^2-m^2)d^{D+1}p. \end{align*}
Let $\omega_m(\tp)=\sqrt{m^2+|\tp|^2}$, then
$\Omega_m=\{(\omega_m(\tp),\tp); \tp\in \RR^{D}\}$ 
and the measure can be expressed in the $\tp$-coordinates as: 
\begin{align}
 d\Omega_m=
 \frac{d^{D}\tp}{\omega_m(\tp)}. \label{eq:omegam}
\end{align}
The action $U_m$ of $(\Lambda,a)\in \poincare = \mathcal{L}^\uparrow_+\ltimes \RR^{D+1}$ is given as follows:
\begin{align}
 (U_m((\Lambda,a))\Psi)(p)=e^{ia\cdot p}\Psi(\Lambda^{-1}p), \text{ for } \Psi \in \hm. \label{poincare-action}
\end{align}

Consider  the restriction $E$ of the Fourier transformation on Schwartz functions on $\RR^{D+1}$ to $\Omega_m:$
\begin{align*}
 E:\S(\RR^{D+1})&\rightarrow \S(\Omega_m) \\
 f&\mapsto (E f)(p)=\int_{\RR^{D+1}} e^{ix\cdot p }f(x) d^{D+1}x, \ \ p \in \Omega_m
\end{align*}
(as $D> 1$, this holds even if $m=0$). Then $E(\S(\RR^{D+1}))$ is dense in $\hm$.
We refer to $p$ as the momentum variable and $x$ as the position variable. We shall denote by $\tx_\perp$ the $D-1$ coordinate vector $(x_2,\ldots,x_D)$.
The action of the Poincar\'e group on the one-particle space (\ref{poincare-action}) is covariant with respect to the action of
$\poincare$ on test functions in $\S(\RR^{D+1})$:
\begin{align*}
 (U_m((\Lambda,a))E f)(x) = E (f(\Lambda^{-1}(x-a))).
\end{align*}

The local structure of the free scalar field is encoded in the local space corresponding to bounded open regions $O \in \RR^{D+1}$:
\begin{align}
 H(O):= \overline{\{Ef \in L^2(\Omega_m,d\Omega_m): f\in \mathcal{S}(\RR^{D+1},\RR), \supp(f)\subset O  \}}.\label{local-subspace}
\end{align}
For an arbitrary regions $S$ with non-empty interior in Minkowski space, its local subspace is generated by the subspaces of bounded open regions contained in it:
\begin{align*}
 H(S)=\overline{\bigcup\limits_{O\in S}H(O)}.
\end{align*}

The map $\RR^{D+1}\supset O\mapsto H(O)\subset \H$  and the Poincar\'e representation $U_m$ define a net of standard subspaces,
also called one-particle net or first quantized net, satisfying the following properties (see e.g.\! \cite{BGL02}):
\begin{enumerate}[{(SS}1{)}]
\item {\bf Isotony:} $H(O_1) \subset H(O_2)$ for $O_1 \subset O_2$;
\item {\bf Locality:} if $O_1\subset O_2'$, then $H(O_1)\subset H(O_2)'$, where $O'$ denotes the spacelike complement of $O$;
\item {\bf Poincar\'e covariance:} $H(gO) = U(g)H(O)$ for 
$g \in {\P}^\uparrow_+$;
\item {\bf Spectral condition:} 
the joint spectrum of the translation subgroup in $U$ is contained in the closed forward light cone
 $\overline {V_+}=\{p\in\RR^{D+1}: p^2\geq0, p_0\geq0\}$.
\item {\bf Cyclicity}: $H(O)$ are cyclic subspaces.
\setcounter{axiomss}{\value{enumi}}
\end{enumerate}
It is a consequence of locality and cyclicity that $H(O)$ are standard subspaces.

Consider the standard wedge $W_1=\{x\in\RR^{D+1}:|x_0|<x_1\}$ in the $x_1$-direction and let
\[
\Lambda_{W_1}(t)(x_0,x_1,\tx_\perp)=(\cosh(t)x_0+\sinh(t)x_1, \sinh(t)x_0+\cosh(t)x_1,\tx_\perp) 
\]
be the one-parameter group of Lorentz boosts fixing $W_1$.
Any other region of the form $W=gW_1, g \in \poincare$ is also called a wedge,
and we put $\Lambda_W(t) = g\Lambda_{W_1}(t)g^{-1}$ with such a $g$
(this is well-defined, because any other $g$ differs only by a Lorentz boost preserving the wedge).
We shall denote by $\mathcal W$ the set of wedges. Let $H(W)=\overline{\bigcup_{O\in W}H(O)}$ be the subspace associated to the wedge $W$.  The net defined by \eqref{local-subspace} further satisfies the following properties. 
\begin{enumerate}[{(SS}1{)}]
\setcounter{enumi}{\value{axiomss}} 
\item {\bf Bisognano--Wichmann (BW) property:} 
Let $W\in\mathcal W$, it holds that $$U(\Lambda_{W}(t)) = \Delta_{H(W)}^{-\frac{it}{2\pi}}\qquad \text{ for } t \in \RR,$$

\item {\bf Haag duality (for wedges):} 
$H(W')=H(W)'$, for all $W\in\mathcal W.$
\setcounter{axiomss}{\value{enumi}}
\end{enumerate}

\subsection{Second quantization and nets of von Neumann algebras}
\label{s:secq} 

Let $\H$ be a Hilbert space and $H\subset\H$ a real linear subspace.
The von Neumann algebra $R(H)$, called second quantization algebras, on the symmetric Fock space
$\mathcal{F}_+(\H)$ generated by the Weyl operators:
\begin{equation*}
R(H) \equiv \{\mathrm{w}(\xi): \xi\in H\}'',
\end{equation*}
where $\mathrm{w}(\xi), \xi \in \H$ are unitary operators on $\mathcal{F}_+(\H)$
characterized by
\[
\mathrm{w}(\xi)e^\eta = e^{-\frac12\<\xi,\xi\>-\<\xi,\eta\>}\cdot e^{\xi+\eta},
\]
where $e^\eta = 1 \oplus \eta \oplus \frac1{2!}\eta^{\otimes 2} \cdots$
is a coherent vector in $\mathcal{F}_+(\H)$.
By strong continuity of the map $\H\ni \xi\mapsto \mathrm{w}(\xi)\in \mathcal{F}_+(\H) $
we have that
\[
R(H) = R(\ol H)\ .
\]
Moreover, the Fock vacuum vector $\Omega = e^0$ is cyclic (respectively separating)
for $R(H)$ if and only if $\ol H$ is cyclic (respectively separating).
Therefore one verifies the relation:
\begin{align}\label{eq:second-quantization-identity}
	\braket{\mathrm{w}(\xi)\Omega,\mathrm{w}(\eta)\Omega}=e^{-\frac{1}{2}(\norm{\xi}^2+\norm{\eta}^2)}e^{\braket{\xi,\eta}}.
\end{align} 

The second quantization respects the lattice structure \cite{Araki63} and the modular structure \cite{LRT78,LMR16}. We shall denote $J_{R(H),\Omega}$, $\Delta_{R(H),\Omega}$ the Tomita operators associated with $(R_+(H),\Omega)$, and by
$\Gamma_+(T)$ the multiplicative second quantization of a one-particle
operator $T$ on $\H$. It holds that $\Gamma_+(T)e^\xi=e^{T\xi}$ for $\xi \in \H$.

\begin{proposition}\label{prop:secquant}\cite{LRT78,LMR16}
Let $H$ and $H_a$ be closed, real linear subspaces of $\H$. We have
\begin{enumerate}
\item $J_{R(H),\Omega} = {\Gamma}_+(J_H)$, \ $\Delta_{R(H),\Omega}= {\Gamma}_+(\Delta_H)$ if $H$ is standard;
\item $R(H)' = R(H')$;
\item $R(\Span_a H_a) = \bigvee_a R(H_a)$;
\item $R(\bigcap_a H_a) = \bigcap_a R(H_a)$,
\end{enumerate}
where $\bigvee_a R(H_a)$ denotes the von Neumann algebra generated by $R(H_a)$.
\end{proposition}

In particular, the second quantization promotes the one-particle net defined in \eqref{local-subspace}
to a Haag-Kastler net of local algebras: consider the map $\RR^{D+1}\supset O\mapsto \A(O)=R(H(O))\subset \mathcal{F}_+(\H)$
together with the second quantization $U=\Gamma_+(U_m)$ of the one-particle Poincar\'e representation $U_m$, then the following hold:
\begin{enumerate}[{(HK}1{)}]
 \item \textbf{Isotony: } $\A(O_1) \subset \A(O_2)$ for $O_1 \subset O_2$; \label{isotony}
 \item \textbf{Locality:}  if $O_1\subset O_2'$, then $\A(O_1)\subset \A(O_2)'$; \label{locality}
 \item\textbf{Poincar\'e covariance:} \label{poincare}
  $U(g)\A(O)U(g)^*=\A(gO)$ for $g\in\P_+^\uparrow$.
 \item\textbf{Positivity of the energy:}
 the joint spectrum of the translation subgroup in $U$ is contained in the closed forward light cone
 $\overline {V_+}$. \label{positiveenergy}
 
 \item \textbf{Vacuum and the Reeh-Schlieder property: } $\Omega$ is the (up to a phase) unique vector
 such that $U(g)\Omega=\Omega$ for $g\in \poincare$ and is cyclic, $\overline{\A(O)\Omega}=\H$ for any $O$. \label{vacuum}
 
 \item\textbf{The Bisognano-Wichmann property:} \label{bw} For a wedge $W \in \mathcal{W}$, 
 we put $\A(W)= \left(\bigvee_{O\subset W} \A(O)\right)''$.
 Then it holds that
 \[
  U(\Lambda_{W}(t))=\Delta_{\A(W),\Omega}^{-\frac{it}{2\pi}},
 \]
 where $\Delta_{\A(W),\Omega}^{it}$ is the modular group of $\A(W)$ with respect to $\Omega$.
 \setcounter{axiomhk}{\value{enumi}}
 \end{enumerate}

\subsection{The \texorpdfstring{$U(1)$}{U(1)}-current net}\label{sec:U(1)-current}
We introduce a family of standard subspaces parametrized by intervals on $\RR$,
the $U(1)$-current. Let us consider
$\H_{\uone} = L^2(\RR,d\theta')$
and, for each $I \subset \RR$ connected open interval, the subspace
\begin{align*}
 H_\uone(I)=\overline{\{\hat{g}(e^{-\theta'});g\in C_0^\infty(\RR,\RR), \hat g(0) = 0, \supp(g)\subset I\}} \subset L^2(\RR,d\theta'),
\end{align*}
where $\hat g$ is the Fourier transform of $g$.
Furthermore, we introduce
\[
(U_{\uone}(\alpha, t)\xi)(\theta') = e^{it e^{-\theta'}}\xi(\theta'-\alpha). 
\]
This is a unitary representation of the translation-dilation group $\mathbf{P} = \RR \ltimes \RR$. Let $x\in\RR$, then $(\alpha,0)\in \RR\ltimes\RR$ corresponds to the dilation $\mathfrak{d}(\alpha) x =e^{\alpha}x$ and $(0,t)\in \RR\ltimes\RR$ corresponds to the translations $\mathfrak{t}(t)x=x+t$. It is straightforward to check that
$U_{\uone}(\alpha, t) H(I) =  H(e^\alpha I + t)$.

Each $H(I)$ is a standard subspace
and for disjoint $I_1, I_2$ it holds that $H(I_1) \subset H(I_2)'$.
Furthermore, the Bisognano-Wichmann property holds:
$\Delta_{H(\RR^+)}^{it}=U_{\uone}(-2\pi t,0), t\in\RR$.

Let $f,g \in C_0^\infty(\RR,\RR)$ with $\hat{f}(0)=\hat{g}(0)=0$ and call the primitives $F,G$, respectively. The condition $\hat{f}(0)=\hat{g}(0)=0$ implies $F,G \in C_0^\infty(\RR,\RR)$.  We have under the substitution $p=e^{-\theta^\prime}$:
	\begin{align*}
		\braket{\hat{g}(e^{-\theta^\prime}),\hat{f}(e^{-\theta^\prime})}&=\int_{\RR_+}\hat{g}(-p)\hat{f}(p)\frac{dp}{p}=\int_{\RR_+}\hat{G}(-p)\hat{F}(p)pdp.
	\end{align*}	
The imaginary part of the scalar product (symplectic form) plays a crucial role in the second quantization  
		\begin{align}\label{eq:U(1)-symplectic-form}
			\im\braket{\hat{g}(e^{-\theta^\prime}),\hat{f}(e^{-\theta^\prime})}&=\frac{1}{2i}\int_{\RR_+}\left(\hat{G}(-p)\hat{F}(p)-\hat{F}(-p)\hat{G}(p)\right)pdp=\frac{1}{2}\int_{\RR} G(x)f(x)dx
		\end{align}	
	
The family $\{H(I)\}$ and the representation $U_{\uone}$ are unitarily equivalent to the family of closed real subspaces coming from the $\uone$-current conformal net and the one-particle symmetry restricted to the translation-dilation group, see \cite[Section 5.2]{BT15}. The intertwining map for $h\in C_0^\infty(\RR,\RR)$ is $\hat{h}(p)\mapsto \widehat{h^\prime}(e^{-\theta'})$. Thus to switch to the standard definition of the $U(1)$-current presented in the literature (see e.g. \cite{Longo08}) one replaces $g\in C_0^\infty(\RR,\RR)$ with $\hat{g}(0)=0$ with its primitive $G\in C_0^\infty(\RR,\RR)$.

\section{Free scalar field on the null plane}\label{free}
\subsection{Direct integrals and decompositions}\label{direct}
Let us summarize some of the basic notions and results on the direct integral of Hilbert spaces and decomposition of group representations.
We follow the conventions of \cite{Dixmier81}, see also \cite{KR97-2}. 

Let $X$ be a $\sigma$-compact locally compact Borel measure space, $\nu$ the completion of a Borel measure on $X$ and $\{\K_\l  \}$ a family of separable Hilbert spaces indexed by $\l\in X$. We say that a separable Hilbert space $\K$ is the \textbf{direct integral} of $\{ \K_\l\}$ over $(X,\nu)$ if, to each $\xi\in \K$ there corresponds a function $\l\mapsto \xi(\l) \in \K_\l$ and 
\begin{enumerate}
\item $\l\mapsto \braket{\xi(\l),\eta(\l)}$ is $\nu$-integrable and $\braket{\xi,\eta}=\int_X \braket{\xi(\l),\eta(\l)} d\nu(\l)$
\item if $\phi_\l \in \K_\l$ for all $\l$ and $\l\mapsto \braket{\phi_\l,\eta(\l)}$ is integrable for all $\eta\in \K$, there is $\phi\in \K$ such that $\phi(\l)=\phi_\l$ for almost every $\l$.
\end{enumerate}
In this case, we write:
\begin{align*}
\K=\int_X^\oplus \K_\l d\nu(\l).
\end{align*}

An operator $T\in \B(\K)$ is said to be \textbf{decomposable} when there is a function $\l\mapsto T_\l $ on $X$ such that $T_\l \in \B(\K_\l)$ and, for each $\xi \in \K$, $T_\l \xi(\l)=(T\xi)(\l)$ for $\nu$-almost every $\l$, and in this case we write
$T = \int_X^\oplus T_\l d\nu(\l)$.
If $T_\l =f(\l)1$ for some $f\in L^\infty(X,\nu)$, we say that $T$ is \textbf{diagonalizable}.
An operator $T\in \B(\K)$ is decomposable if and only if $T$ commutes with every diagonalizable operator.
Conversely,
let $X\ni\lambda\rightarrow T_\lambda\in \B(\K_\lambda)$ be a field of bounded operators
with $\sup_\lambda \|T_\lambda\| < \infty$.
If for any $\xi\in\K$ with $\xi=\int_X^\oplus \xi(\lambda) d\nu(\lambda)$ there exists $\eta\in\K$ such that $\eta(\lambda)=T_\lambda\xi(\lambda)$ for almost every $\lambda\in X$, then $T=\int_X^\oplus T_\lambda d\nu(\lambda)\in\B(\K)$ defines a bounded decomposable operator on $\K$ and $\lambda\rightarrow T_\lambda$ is called a \textit{measurable field of bounded operators}.

The following Lemma can be found in \cite[Appendix B]{MT18}.
\begin{lemma}\label{lem:ABC}
Let $\K=\int_X^\oplus \K_\lambda d\nu(\lambda)$ then
\begin{itemize}
\item[(a)]
Let  $H=\int_X^{\oplus_\RR} H_\lambda d\nu(\lambda)\subset \K$ be a real subspace and the direct integral is taken on $\RR$
such that   $H_\lambda\subset \K_\lambda$,
then $H'=\int_X^{\oplus_\RR} H'_\lambda d\nu(\lambda)$.

\item[(b)] Let $\{H_k\}_{k\in\NN}$ be a countable family of real subspaces of $\K$ such that  $H_k=\int_X^{\oplus_\RR} (H_k)_\lambda d\nu(\lambda)$ on $\RR$
and  $(H_k)_\lambda\subset \K_\lambda$, then  $\bigcap_{k\in\NN} H_k=\int_X^{\oplus_\RR} \bigcap_{k\in\NN} (H_k)_\l d\nu(\l)$.

\item[(c)]  Let $\{H_k\}_{k\in\NN}$ be a countable family of real subspaces of $\K$ such that  $H_k=\int_X^{\oplus_\RR} (H_k)_\lambda d\nu(\lambda)$ on $\RR$
and  $(H_k)_\lambda\subset \K_\lambda$, then  $\overline{\Span_{k\in\NN} H_k}=\int_X^{\oplus_\RR} \overline{\Span_{k\in\NN} (H_k) }_\l d\nu(\l)$.
\end{itemize}
\end{lemma}
Let $G$ be a locally compact group and $\pi$ a continuous unitary representation of $G$ on $\K=\int^\oplus_X \K_\l d\nu(\lambda)$. Suppose that, for each $g\in G$, we have $\pi(g)=\int^\oplus \pi_\l (g) d\mu(\lambda)$, then we say that $\pi$ is the \textbf{direct integral} of the $\pi_\l $ and write: 
\begin{align*}
\pi=\int^\oplus_X \pi_\l d\nu(\lambda).
\end{align*}
Equivalently, $\pi$ is a direct integral if each $\pi(g), \ g\in G$, is decomposable.
It follows immediately that
the direct integral of standard subspaces is again a standard subspaces on the direct integral of Hilbert space,
and so is a direct integral of Poincar\'e covariant nets of standard subspaces.

As a particular case, let $\mathcal{K}= \int_X^\oplus \mathcal{K}_0 d\nu$ be a direct integral Hilbert space over the constant field $\mathcal{K}_0$,
which is a separable Hilbert space, on $X$ with measure $\nu$. Then it holds that \cite[Proposition II.1.8.11, Corollary]{Dixmier81}
\begin{align}\label{integral=tensor-prod}
\int_X^\oplus \mathcal{K}_0 d\nu(\l) \simeq L^2(X,\nu)\otimes \mathcal{K}_0.
\end{align}
The isomorphism is given by identifying $L^2(X,\nu)\otimes \K_0$ as the space of $\K_0$-valued $L^2$-functions.
By this isomorphism, we identify $\int_X^\oplus T d\nu(\l)$ and $\1\otimes T$, where $T\in \mathcal{B}(\mathcal{K}_0)$
(a constant field of bounded operators).

\subsection{Decomposition of the one-particle space}\label{section-decomposition-one-particle-space}
Let us fix $D>1$. We first consider the $m>0$ case. In Remark \ref{rem:massless} we explain the minor modification to deal with the massless case.

The hypersurface $X_-^0 := \{x\in\RR^{D+1}:x_0-x_1=0\}$ is called the \textbf{null plane} in the $x_1$-direction.
It is appropriate to consider the coordinate frame $(x_+,x_-,\tx_\perp)=(\frac{x_0+x_1}{\sqrt 2}, \frac{x_0-x_1}{\sqrt 2}, \tx_\perp)$. 
In these coordinates, the Minkoswki product becomes
\begin{align}\label{eq:scprod}
 x\cdot p= x_+p_-+x_-p_+-\tx_\perp \tp_\perp
\end{align}
 and the Minkowski (pseudo)norm $x^2=2x_+x_--\tx_\perp^2$.
 In the momentum space, the massive hyperboloid is determined by $2p_+p_- -\tp_\perp^2= m^2$,
 and for each $(p_-, \tp_\perp) \in \RR_+ \times \RR^{D-1}$ there is one and only one $p_+ \in \RR_+$
 satisfying this equation.
 For a test function $f=f(x_+,x_-,\tx_\perp)$ on $\RR^{D+1}$,
 let $E f$ be its Fourier transform restricted to $\Omega_m= \left\{ p=\big({\frac{m^2+\tp_\perp^2}{2p_-}}, p_-,\tp_\perp\big): p_->0, \tp\in\RR^{D-1}\right\}$ (in the $(p_+,p_-,\tp_\perp)$-coordinates).
 In the $(p_-,\tp_\perp)$-coordinates,
 we have, up to a unitary (given by the change of variables),
 \[
  \hm \simeq L^2\left(\RR^D,  \frac{d^{D}\tp}{\omega_m(\tp)}\right) \simeq
  L^2\left(\RR_+ \times \RR^{D-1}, \frac{2p_-dp_-d\tp_\perp}{m^2+2p_-^2+\tp_\perp^2}\right)  
 \]
 (since $dp_1d\tp_\perp=\frac1{\sqrt2}dp_-\tp_\perp$ and $\omega_m (p_-,\tp_\perp)=\frac1{\sqrt2}\left(\frac{m^2+\tp_\perp^2}{2p_-}+p_-\right)$ is the dispersion relation in the $(p_-,\tp_\perp)$-coordinates).
 For sake of notational simplicity, we will write $\xi(p_1,\tp_\perp)\in \hm$ or $\xi(p_-,\tp_\perp) \in\hm$ for $\xi \in\hm$
 in terms of $(p_1,\tp_\perp)$ or $(p_-,\tp_\perp)$, respectively.
 In the same way, the representation $U_m$ acts on $\hm$ in various realizations.
 
  We introduce the map $V_{\mathrm M}:\hm = L^2(\Omega_m,d\Omega_m) \to L^2(\RR^{D},d\theta'd^{D-1}\tp_\perp)$ as follows:
  \begin{align*}
  (V_{\mathrm M}^{-1}\xi)(\theta',\tp_\perp) = \textstyle{\xi(\mathrm{\arcsinh}(\frac{p_1}{\omega_m(\tp_\perp)}))-\ln(\omega_m(\tp_\perp))+\ln\sqrt 2, \tp_\perp)},\quad \xi\in L^2(\RR^D,d\theta'd^{D-1}\tp_\perp),
  \end{align*}
  where $\omega_m(\tp_\perp) := \sqrt{m^2 + \tp_\perp^2}$.
$  V_\mathrm{M}$ is a  unitary operator. Indeed, first,
  the change of $p_1$ to rapidity $\theta =\mathrm{\arcsinh}(\frac{p_1}{\omega_m(\tp_\perp)})$ for a fixed $\tp_\perp$, or
    $p_1 = \omega_m(\tp_\perp)\sinh \theta$ is a smooth one-to-one map $\RR^D\rightarrow \RR^D$.
    Moreover, with $\frac{\partial \theta}{\partial p_1} = \frac{1}{\omega_m(\tp_\perp)}/\sqrt{1+\frac{p_1^2}{\omega_m(\tp_\perp)^2}} = \frac1{\omega_m(\tp)}$, the measure $d\Omega_m$ transforms in the following way, cf.\! \eqref{eq:omegam}:
     \begin{align*}
      d\Omega_m&=\frac{d^{D}\tp}{\omega_m(\tp)}=d\theta d^{D-1}\tp_\perp.
     \end{align*}
    Therefore, the pullback of the rapidity substitution imposes the equivalence:
     \begin{align*}
      \hm=L^2(\Omega_m,d\Omega_m)\simeq L^2(\RR\times \RR^{D-1},d\theta \times d^{D-1}\tp_\perp) (\simeq L^2(\RR,d\theta)\otimes L^2(\RR^{D-1},d^{D-1}\tp_\perp)).
     \end{align*}
    Moreover, the substitution 
    $\theta'=\theta - \ln(\omega_m(\tp_\perp))+\ln\sqrt 2$ is a smooth one-to-one map: $\RR \rightarrow \RR$ if $m>0$ with $d\theta=d\theta'$.

By substituting $\theta = \arcsinh(\frac{p_1}{\omega_m(\tp_\perp)})$ or $p_1 = \omega_m(\tp_\perp)\sinh \theta$
in $p_- = \frac1{\sqrt 2}(\omega_m(p_1,\tp_\perp) - p_1)$
we have $p_-= \frac1{\sqrt 2}\omega_m(\tp_\perp)e^{-\theta}$
therefore, 
\[
p_-= \frac1{\sqrt 2}\omega_m(\tp_\perp)e^{-\theta'-\ln(\omega_m(\tp_\perp))+\ln\sqrt 2}=e^{-\theta'},
 \]
and 
\[
(V_{\mathrm M} \xi)(\theta',\tp_\perp)=\xi(e^{-\theta'}, \tp_\perp), \quad \xi\in L^2\left(\RR_+ \times \RR^{D-1}, \frac{2p_-dp_-d\tp_\perp}{m^2+2p_-^2+\tp_\perp^2}\right) 
\]

 \begin{proposition}
    We have the following equivalence, where the map is given by $V_{\mathrm M}$ composed by the inverse Fourier transform on
    the perpendicular momenta:
    \begin{align}\label{diagram-decompositions}
       \hm\simeq L^2(\RR^{D},d\theta'd^{D-1}\pmb{x}_\perp) \simeq \int_{\RR^{D-1}}^\oplus L^2(\RR,d\theta')d^{D-1}\pmb{x}_\perp,
     \end{align}
   where the last equivalence is given by \eqref{integral=tensor-prod}.
  \end{proposition}
Subsequently, we will refer to these direct integral representations of $\hm$ as \textbf{spatial decomposition}
and denote the intertwining unitary by $V_{\mathrm S}$,
while the map $V_{\mathrm M}$ is referred to as \textbf{momentum decomposition}.

Next we discuss the decomposition of the spacetime symmetry.
Let us take the wedge $W_1$ in the $x_1$-direction and consider the subgroup
$\Lambda_{1}\ltimes \mathfrak{t}_{x_+}\subset  \poincare$ consisting of boosts $\Lambda_{1}:=\Lambda_{W_1}$ along the $x_1$-direction and
lightlike translations $\mathfrak{t}_{x_+}$ of $x_+$.
It is straightforward that this group is isomorphic to the translation-dilation group $\mathbf{P}$ of $\RR$.
We will show that the restriction of $U_m$ to $\mathbf{P}$ decomposes with respect to the decompositions
\eqref{diagram-decompositions} of the one-particle vectors.

\begin{lemma}\label{action-momentum-decomp}
 The unitary $V_{\mathrm M}$ between $\hm$ and $\int_{\RR^{D-1}}^\oplus L^2(\RR,d\theta') d^{D-1}\tp_\perp$
 intertwines the actions of the subgroup $\mathbf{P}$ to:
 \begin{align}\label{eq:U_m-momentumdecomp}
(V_{\mathrm M} U_m(\Lambda_{1}(\alpha),\mathfrak{t}_{x_+}(a_+))\xi)(\theta',\tp_\perp)&= e^{i  a_+e^{-\theta'} }(V_{\mathrm M}\xi)(\theta'-\alpha,\tp_\perp),
 \end{align}
\end{lemma}
\begin{proof}
We start with the Poincar\'e group action given by $U_m$ on $\hm$ expressed in \eqref{poincare-action}.
We recall that $p_-= e^{-\theta'}$. 
Then the result is immediate by considering \eqref{eq:scprod} and the fact that $a_+\cdot p_- =a_+e^{-\theta'} $,

For boosts, the claim follows directly from the following
       \begin{align*}
     \Lambda_{1}^{-1}(\a)p&=\left(\begin{array}{c}
     \cosh(\alpha )p_0 - \sinh(\alpha)p_1 \\
     -\sinh(\alpha)p_0+\cosh(\alpha)p_1\\
     \tp_\perp
     \end{array}  \right) \\
     &= \left(\begin{array}{c}
     \omega_m(\tp_\perp) \cosh(\theta'+\ln(\omega_m(\tp_\perp)) - \ln\sqrt 2- \alpha) \\
     \omega_m(\tp_\perp)\sinh(\theta'+\ln(\omega_m(\tp_\perp))- \ln\sqrt 2-\alpha)\\
     \tp_\perp
     \end{array}  \right).\nonumber \\
        \end{align*}
  \end{proof}

 \begin{proposition}\label{spatial-decomp-P}
 Under the spatial decomposition, the representation of $U_m$ to $\mathbf{P}$ is decomposable (in the sense of Section \ref{direct})
 and acts as follows:
    \begin{align}\label{eq:U_m-spatialdecomp}
    (U_m((\Lambda_{1}(\alpha),\mathfrak{t}_{x_+}(a_+)))\xi)(\theta',\pmb{x}_\perp) &=
    e^{ia_+e^{-\theta'}}\xi(\theta'-\alpha, \pmb{x}_\perp).
   \end{align}
 \end{proposition}
 \begin{proof}
  The representation $U_m$ decomposes in the momentum decomposition as \eqref{eq:U_m-momentumdecomp},
  and the action of $U_m$ does not involve the variable $\tp_\perp$,
  it commutes with the inverse Fourier transform on $\tp_\perp$.
 \end{proof}

\begin{remark}\label{rem:massless} The disintegration \eqref{diagram-decompositions} and the identifications \eqref{eq:U_m-momentumdecomp} and \eqref{eq:U_m-spatialdecomp} apply also to the $m=0$ case up to a measure zero set. 
To see this, note that in the massless case $\Omega_0=\partial V_+=\{p\in\RR^{D+1}: p^2=0,p_0>0\}$. The change of variables determining  $V_{\mathrm M}^{-1}$ (both in $(p_1,\tp_\perp)$ or $(p_-,\tp_\perp)$ set of coordinates) are 1-1 diffeomorphisms on the set $\Omega_0\cap\{p\in\RR^{D+1}: \tp_\perp\neq 0\}$. This set is dense in $\Omega_0$ with respect to the topology induced by the euclidean topology on $\RR^{D+1}$. The complement   $N_0=\{p\in\Omega_0:\tp_\perp= 0\}$ has measure zero
with respect to  the Lorentz invariant measure on $\Omega_0$, see \eqref{eq:omegam}.  We conclude that $V_{\mathrm M}$ is a  unitary operator.
As $\Lambda_{1}$-boosts  fix $N_0$,
\eqref{diagram-decompositions}, \eqref{eq:U_m-momentumdecomp} and \eqref{eq:U_m-spatialdecomp} continue to hold (up to a measure zero set).
This is the only modification to have in mind along the paper to adapt the massless case.
 \end{remark}

\subsection{The local subspace of the null plane}\label{sec:local-subspace-of-null-plane}

 Here we show that the free field can be restricted to regions on the null plane if we exclude the lightlike zero-modes
 (see Lemma \ref{lm:zero} for the precise restriction, cf.\! \cite{Ullrich04}).
 Recall that the algebra $\A(O)$ is generated by the Weyl operators associated with test functions supported in $O$.
 In this spirit, we define the local subspaces of open (with the relative topology) regions $R$ in the null plane $X_-^0:=\{x\in \RR^{D+1};x_-=0\}$. 
 To make it precise, we consider distributions of the form 
  \begin{align}
    g_0(x_+,x_-,\pmb{x}_\perp)=\delta(x_-)g(x_+,\pmb{x}_\perp), \label{simple-test-function-on-null-plane}
   \end{align}
  where $g$ is a real-valued function such  that $\RR\ni x_+\mapsto g(x_+,\tx_\perp)$ is a compactly supported smooth function on $X_-^0$.
  We denote the set of such distributions by $\mathscr{D}(X_-^0,\RR)$,
  and call them \textbf{thin test functions} supported on $X_-^0$.

 The Fourier transform of such $g_0$ is a smooth function and can be restricted to the mass shell,
 which we denote by $Eg_0$.
 Then it is natural to consider one-particle vectors associated with thin test functions.
    \begin{lemma}\label{lm:zero}
  Let $g_0(x_+,x_-,\pmb{x}_\perp)= \delta(x_-)g(x_+, \pmb{x}_\perp)$ as in \eqref{simple-test-function-on-null-plane}.
  Then $E g_0 \in L^2(\Omega_m, d\Omega_m)$ (the norm is finite) if and only if
 $\int dx_+ g(x_+, \pmb{x}_\perp) = 0$ for each $\pmb{x}_\perp$.
  \end{lemma}
  \begin{proof}
    We take $g_0(x_+,x_-,\pmb{x}_\perp)= \delta(x_-)g(x_+, \pmb{x}_\perp)$ as above. 
   The Fourier transform of $g_0$ is
     \begin{align*}
       (E g_0)(p)&=\int_{\RR^{D+1}} e^{i(x_+p_-+x_-p_+-\pmb{x}_\perp\cdot\tp_\perp)}\delta(x_-)g(x_+, \pmb{x}_\perp)dx_+dx_-d^{D-1}\pmb{x}_\perp=\hat g(p_-, \tp_\perp).
     \end{align*}
    By rapidity substitution $p_- =\frac1{\sqrt 2}\omega_m(\tp_\perp)e^{-\theta}$ and $p_+=\frac1{\sqrt 2}\omega_m(\tp_\perp)e^{\theta}$ we have: 
     \begin{align*}
       \norm{E g_0}^2_{\hm^-}&=\int_{\RR\times \RR^{D-1}} \left|\hat g\left(\frac1{\sqrt 2}\omega_m(\tp_\perp)e^{-\theta}, \tp_\perp\right)\right|^2d\theta d^{D-1}\tp_\perp
       =\int_{\RR\times \RR^{D-1}} |\hat g(e^{-\theta'}, \tp_\perp)|^2d\theta' d^{D-1}\tp_\perp
     \end{align*}
    as we did in \eqref{diagram-decompositions}.
    Therefore,
    it is necessary that $\hat{g}(0,\tp_\perp) = 0$ for each $\tp_\perp$ in order to have $E g_0\in \hm$ (that is the integral is finite),
    and equivalently,  $\int dx_+ g(x_+, \pmb{x}_\perp) = 0$ for each $\pmb{x}_\perp$.

    Conversely, if $\int dx_+ g(x_+, \pmb{x}_\perp) = 0$ for each $\pmb{x}_\perp$, then $\hat g(0, \tp_\perp)= 0$. In particular, since $g$ is smooth by the Taylor formula $\theta'\mapsto g(e^{-\theta'},\tp_\perp)$ has finite $L^2$-norm at $+\infty$. Then since $g$ is  compactly supported, then it decays fast when $p_-\rightarrow-\infty$ and $|\tp_\perp|\rightarrow+\infty$. As a consequence $E g_0\in \hm$.
\end{proof}

The condition $\int dx_+ g(x_+, \pmb{x}_\perp) = 0$ for each $\pmb{x}_\perp$
is satisfied if $g$ is the partial derivative in $x_+$ of another compactly supported smooth function.
We define, for an open bounded region $R$ on the null plane $X^0_-$, as follows:
 \begin{align*}
  H(R):=\overline{\Span}^{\hm}{\{E g_0 \in L^2(\Omega_m,d\Omega_m):g_0=\delta(x_-)\partial_{x_+}g(x_+,\pmb{x}_\perp)\in \mathscr{D}(X_-^0,\RR), \supp g \subset R\}}.
 \end{align*}
   For a general open region $\tilde R\subset X_-^0$, the subspace is generated by the subspaces of open, bounded regions inside $\tilde R$:
   \begin{align*}
    H(\tilde R)=\overline{\bigcup_{R\in \tilde R}H(R)}.
   \end{align*}
\begin{figure}[t]
	\centering
	\includegraphics[width=.6\textwidth]{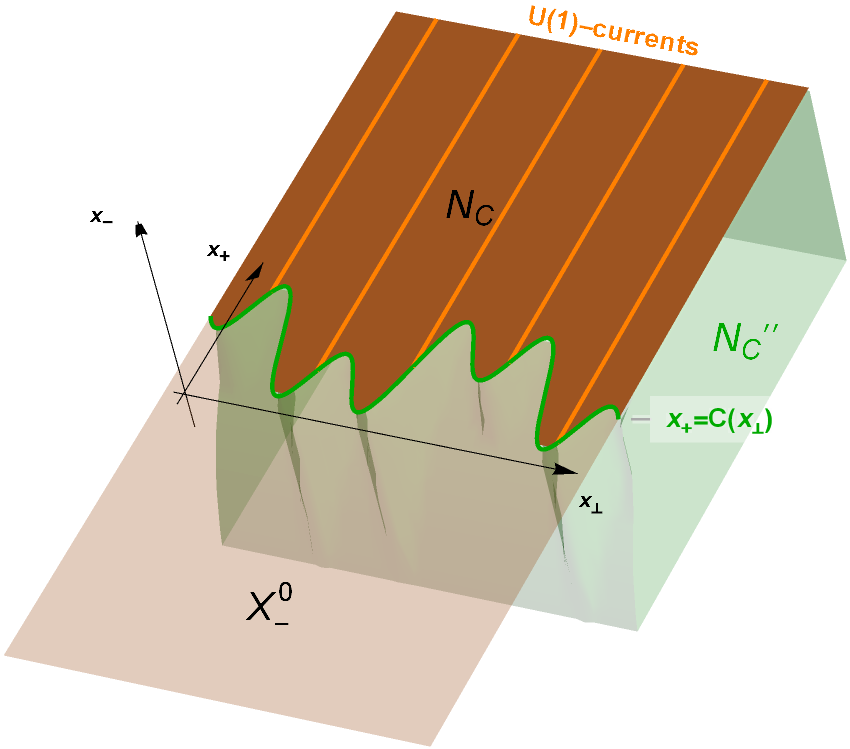}
	\caption{The geometric setup and illustration of the decomposition of standard subspaces of null cuts.}
	\label{fig:setup}
\end{figure}

In particular, a \textbf{null cut} is a region on the null plane associated with
a continuous function $C:\RR^{D-1}\rightarrow \RR$, where $\RR^{D-1}$ denotes the subspace $\{\pmb{x}\in \RR^{D+1};x_-=x_+=0\}$
  \begin{align*}
    N_C:=\{\pmb{x}=(x_+,x_-,\tx_\perp) \in \RR^{D+1};x_-=0,x_+>C(\pmb{x}_\perp) \}.
   \end{align*}

These subspaces are natural in the view of the following isotony in an extended sense.
\begin{lemma}\label{lm:isotony-null-space-region}
 Let $R\subset X_-^0$ be open (relatively in $X_-^0$) and $O\subset \RR^{D+1}$ open such that
 there is a direction $a \in \RR^{D+1}$ where $R\subset O + ta$ with  $0<t<\epsilon$ for sufficiently small $\epsilon$.
 Then it holds that $H(R)\subset H(O)$.
\end{lemma}
\begin{proof}
 Let $g$ be any smooth function compactly supported in $R$. Then for $g_-$ a test function on $R$
 such that $g_-(t) \ge 0$ and $\int g_-(t)dt = 1$, $g(x_-, x_+, \pmb{x}_\perp) = g_-(x_-)\partial_{x_+}g(x_+, \pmb{x}_\perp)$
 has a Fourier transform in $\hm = L^2(\Omega_m,d\Omega_m)$.
 Indeed, we know that the Fourier transform of $\partial_{x_+}g(x_+, \pmb{x}_\perp)$ alone is in $\hm$,
 and it gets multiplied by $\hat g_-(p_+)$, which is rapidly decreasing.
 If we consider the scaled function $g_{-,n}(x_-) = ng_-(nx_-)$, its Fourier transform is $\hat g_-(\frac{p_+}{n})$,
 which is bounded and converges to $1$, therefore, $E(g_{-,n}g)$ converges to $Eg_0$.

 On the other hand, for fixed $a$ and $t$ as in the statement, for sufficiently large $n$, $g_{-,n}g$ is supported in $O + ta$,
 hence we have $Eg_0 \subset H(O+ta)$ by the closedness of $H(O+ta)$.
 This implies $H(R) \subset H(O+ta) = U_m(ta,0)H(O)$, and by the continuity of $U_m$, we obtain $H(R) \subset H(O)$.
\end{proof}

The spaces $H(R)$ are real subspaces of $\hm$.
Recalling the spatial decomposition \eqref{diagram-decompositions} of $\hm$, we will show that the same decomposition holds as a real Hilbert space through $V_{\mathrm S}$.

Let $I \subset \RR$ a bounded open interval and $I_\perp \subset \RR^{D-1}$ a bounded open region.
Under the spatial decomposition \eqref{diagram-decompositions}, we consider the local subspace $H(R)$ of rectangular regions as follows
\[
 R = \{(x_+,x_-,\pmb{x}_\perp): x_-=0, x_+ \in I, \pmb{x}_\perp \in I_\perp\}. 
\]

  \begin{proposition}\label{decomposition-subspace-null-product}
  Let $R$ be a rectangle on the null plane $X_-^0$.
   Under the spatial decomposition \eqref{diagram-decompositions}, the local subspace $H(R)$
   decomposes in the following way (see e.g.\! \cite[Proposition B.2]{MT18}):
     \begin{align}\label{eq:disrect}
      H(R) &\to \int^{\oplus_\RR}_{I_\perp} H_{U(1)}(I) d^{D-1}\pmb{x}_\perp \\
      Eg_0 &\mapsto  V_{\mathrm S}Eg_0,\nonumber
     \end{align}
 where $g_0(x_-,x_+,\pmb{x}_\perp)= \delta(x_-)g(x_+, \pmb{x}_\perp)$ as in \eqref{simple-test-function-on-null-plane}
 satisfying the condition of Lemma \ref{lm:zero}.

  $H(R)$ is separating for any $R = I\times I_\perp$ where $I$ is proper,
  and standard if $I$ contains an open interval and $I_\perp = \RR^{D-1}$.
  \end{proposition}
  \begin{proof}
    By the spatial decomposition \eqref{diagram-decompositions},
    we identify the whole Hilbert space $\hm$ with $\int_{\RR^{D-1}} L^2(\RR,d\theta')d\pmb{x}_\perp$,
    and hence by Lemma \ref{lm:zero} and \eqref{eq:U_m-spatialdecomp} with the constant field $\{H_{U(1)}(I)\}$ of real Hilbert spaces, the right-hand side
    is identified with $H_{U(1)}(I)\otimes_\RR L^2(I_\perp,\RR, d\pmb{x}_\perp)$.
    On the other hand, $H(R)$ is generated by the functions of the form
    $\hat g_+(e^{-\theta'})g_\perp(\pmb{x}_\perp)$, where $\supp g_+ \subset I, \hat g_+(0) = 0$ and $\supp g_\perp \subset I_\perp$,
    hence they coincide.
    
    The statement about the separating property and cyclicity follows from this decomposition and Lemma \ref{lem:ABC}.
    
\end{proof}

Next let us consider translations and dilations on each fibre on the null plane.
Let $C:\RR^{D-1}\to\RR$ be a continuous function.
We define the distorted lightlike translations as the following unitary operator on $ L^2(\RR^D ,d\theta'd\pmb{x}_\perp)$
(and with the identification of $\hm$ with it through the spacial decomposition $V_{\mathrm S}$):
   \begin{align}\label{eq:non-const-translation}
    (T_C \xi)(\theta',\pmb{x}_\perp) = e^{iC(\pmb{x}_\perp)e^{-\theta'}} \xi(\theta', \pmb{x}_\perp).
   \end{align}

Similarly, we consider the distorted lightlike dilations
   \begin{align}\label{eq:non-const-dilations}
    (D_C\xi)(\theta', \pmb{x}_\perp) = \xi(\theta' - C(\pmb{x}_\perp), \pmb{x}_\perp).
   \end{align}
If $C = \alpha$ is constant, they coincide with the usual translation and the dilation, respectively, see \eqref{eq:U_m-spatialdecomp}.

\begin{remark}\label{rem:TCDC}Note that since $C$ is a continuous function $C$, $x_\perp\mapsto  e^{iC(\pmb{x}_\perp)te^{-\theta'}}$ is a measurable vector field of bounded operators,
that is, $\int_{\RR^{D-1}}^\oplus e^{iC(\pmb{x}_\perp)te^{-\theta'}}\xi(\pmb{x}_\perp)d\pmb x_\perp\in\hm$ for all $\xi\in\hm$, with the identification \eqref{diagram-decompositions}.
As an intermediate step, if $C$ is a characteristic function, the claim is obvious.
For an arbitrary continuous $C$, there is a family of simple functions $\{C_N\}_{N\in\NN}$ converging  uniformly to pointwise to $C$.
 Then $e^{iC_N(\pmb{x}_\perp)te^{-\theta'}}$ converges to $e^{iC(\pmb{x}_\perp)te^{-\theta'}}$ in the strong operator topology. 
 $D_C$ is also a decomposable operator since
  $(D_C\xi)(\theta',\tx_\perp)=\int_{\RR^{D-1}}^\oplus D_C(\tx_\perp) \xi(\theta',\tx_\perp)d\tx_\perp$,
 where $D_C(\tx_\perp) \xi(\theta',\tx_\perp)= \xi(\theta' - C(\pmb{x}_\perp), \pmb{x}_\perp)$,
 and one can analogously prove that $x_\perp\mapsto  D_C(\tx_\perp)$ is a measurable vector field of bounded operators passing to the $\theta'$-Fourier transform. In particular, since $H=\int_{\RR^{D-1}}^{\oplus_\RR} H(\tx_\perp)d\tx_\perp\subset\hm$ is a standard subspace,
 the subspaces $T_C H=\int_{\RR^{D-1}}^{\oplus_\RR} T_C(\tx_\perp)H(\tx_\perp)d\tx_\perp$ and $D_CH=\int_{\RR^{D-1}}^\oplus D_C(\tx_\perp)H(\tx_\perp)d\tx_\perp$
 are well-defined and standard as well.
\end{remark}
 
\begin{proposition}\label{pr:null-covariance}
The family of real subspaces $H(R)$, indexed by open connected regions $R\subset X_-^0$, is covariant with respect to
$T_C, D_C$ for a continuous function $C$,
i.e.\! $T_C H(R) = H(R+C)$ and $D_{C} H(R) = H(e^{C}\cdot R)$,
where $R + C = \{(x_+ + C(\pmb{x}_\perp), \pmb{x}_\perp): (x_+, \pmb{x}_\perp) \in R\}$
and $C\cdot R = \{(e^{C(\pmb{x}_\perp)}x_+, \pmb{x}_\perp): (x_+, \pmb{x}_\perp) \in R\}$.
\end{proposition}
\begin{proof}
 Let $g_0$ be a thin test function as in \eqref{simple-test-function-on-null-plane}.
 The fibre-wise deformed test function $g_{\mathfrak{t}_C}(x_+, \pmb{x}_\perp) = g(x_+ - C(\pmb{x}_\perp), \pmb{x}_\perp)$
 corresponds to the following one-particle vector:
 \[
  \hat g_{\mathfrak{t}_C}(p_-, \tp_\perp) = e^{iC(\pmb{x}_\perp)p_-}(\hat g)(p_-, \tp_\perp),
 \]
 Therefore, we obtain the equality $ V_{\mathrm S}(E g_{0,\mathfrak{t}_C})=T_CV_{\mathrm S}(Eg_0) $
 once we prove that $H(R+C)$ contains the vectors which are the Fourier transform
 of $g_{0,\mathfrak{t}_C}(x_+, x_-, \pmb{x}_\perp) = \delta(x_-)g_{\mathfrak{t}_C}(x_+, \pmb{x}_\perp)$
 where $g_{\mathfrak{t}_C}$ is of the form $g(x_+ - C(\pmb{x}_\perp), \pmb{x}_\perp)$
 with a smooth $g$.

 Indeed, we can mollify $g_{\mathfrak{t}_C}$ by an approximate delta function $\delta_n$ on $X_-^0$:
 $\delta_n$ is a positive smooth function with $\int \delta_1(x_+,\pmb{x}_\perp) dx_+ d\pmb{x}_\perp = 1$
 over the Lebesgue measure and $\delta_n(x_+,\pmb{x}_\perp) = n^D \delta_1(nx_+, n\pmb{x}_\perp)$.
 The mollified function $g_{\mathfrak{t}_C} * \delta_n$ is smooth and supported in $R+C$ for sufficiently large $n$,
 its $x_-$ integral vanishes.
 Therefore, its Fourier transform belongs to $H(R+C)$ when restricted to the mass shell,
 and is equal to the Fourier transform of $g_{\mathfrak{t}_C}$ multiplied by that of $\delta_n$. 
 The Fourier transform of $\delta_n$ is bounded and converge to $1$, therefore, it converges to $Eg_{0,\mathfrak{t}_C}$
 where $g_{0,\mathfrak{t}_C}(x_+, x_-, \pmb{x}_\perp) = \delta(x_-)g_{0,\mathfrak{t}_C}(x_+, \pmb{x}_\perp)$.
 As $H(R+C)$ is closed, it contains $Eg_{0,\mathfrak{t}_C}$.

Similarly, for $g_{\mathfrak{d}_{C}}(x_+, \pmb{x}_\perp) = g(e^{-C(\pmb{x}_\perp)}x_+ , \pmb{x}_\perp)$,
 its Fourier transform on $x_+$ for fixed $\pmb{x}_\perp$ as a function of $p_-$ is dilated by $e^{C(\pmb{x}_\perp)}$ and hence
 \begin{align*}
  (V_{\mathrm S}\hat g_{\mathfrak{d}_{C}})(\theta', \pmb{x}_\perp)
  = (V_{\mathrm S}\hat g)(\theta'-C(\pmb{x}_\perp), \pmb{x}_\perp),
 \end{align*}
 that is, it holds that $V_{\mathrm S}(Eg_{0,\mathfrak{d}_{C}}) = D_{C} V_{\mathrm S}(Eg_0)$.
 Similarly as in the case of $g_{\mathfrak{t}_C}$, $Eg_{0,\mathfrak{d}_C} \in H(e^{C}\cdot R)$,
 where $g_{0,\mathfrak{d}_C}(x_+, x_-, \pmb{x}_\perp) = \delta(x_-)g_{\mathfrak{d}_{C}}(x_+, \pmb{x}_\perp)$,
 we obtain the desired covariance.

\end{proof}

Proposition \ref{pr:null-covariance} allows to generalize the disintegration \eqref{eq:disrect} to more general open regions.

\begin{proposition}\label{decomposition-subspace-null-general}
Let $R$ be a region on $X_-^0$ of the form $R = \{(x_+, \pmb{x}_\perp): \pmb{x}_\perp \in I_\perp, C_1(\pmb{x}_\perp) < x_+ < C_2(\pmb{x}_\perp)\}$
for an open region $I_\perp \subset \RR^{D-1}$ and two continuous functions $C_1, C_2$.
Then through the spatial decomposition we have
\[
  H(R) \simeq \int^{\oplus_\RR}_{I_\perp} H_{U(1)}\left((C_1(\pmb{x}_\perp),C_2(\pmb{x}_\perp))\right)d^{D-1}\pmb{x}_\perp.
\]
Similarly, the subspace of a null-cut decomposes:
	\begin{align*}
		H(N_C)&\simeq\int^{\oplus_\RR}_{\RR^{D-1}}H_{U(1)}((C(\pmb{x}_\perp),\infty))d^{D-1}\pmb{x}_\perp, \\
		H(N_{C}^\dagger)&\simeq\int^{\oplus_\RR}_{\RR^{D-1}} H_{U(1)}\left( (-\infty, C(\pmb{x}_\perp))  \right)d^{D-1}\pmb{x}_\perp.
	\end{align*}
\end{proposition}
\begin{proof}
Recall that by Proposition \ref{decomposition-subspace-null-product} the decomposition holds
for the case $R = (0,1)\times I_\perp$.
We can dilate this by $D_{C_2-C_1}$ and then translate it by $C_1$ to arrive at $R$ in the desired form.
By the covariance of the whole family $\{H(R)\}$ with respect to fibre-wise translations and dilations by Proposition \ref{pr:null-covariance},
we obtain the desired decomposition.

If we start from the rectangle $\RR_{\pm}\times \RR^{D-1}$, then we deform the rectangle via the distorted lightlike translation $T_C$ to the null-cut $N_C$
or the interior of its complement $N_C^\dagger$, respectively. The interior of the complement $N_{C}^\dagger$ of the null-cut $N_{C}$ on the null-plane is:
	\begin{align*}
		N_C^\dagger&=X_-^0\setminus \overline{N_C} =\{ x_+<C(\pmb{x}_\perp),x_-=0, \pmb{x}_\perp\in\mathbb{R}^{D-1}  \}.
	\end{align*}
From the covariance of $T_C$ and \autoref{decomposition-subspace-null-product}, we obtain the decompositions:
	\begin{align}
		H(N_{C})&\simeq \int^{\oplus_\RR}_{\RR^{D-1}}H_{U(1)}((C(\pmb{x}_\perp),\infty))d^{D-1}\pmb{x}_\perp\nonumber\\
		H(N_{C}^\dagger)&\simeq\int^{\oplus_\RR}_{\RR^{D-1}} H_{U(1)}\left( (-\infty, C(\pmb{x}_\perp))  \right)d^{D-1}\pmb{x}_\perp\label{decomposition-subspace-complement-null-cut}.
	\end{align}
\end{proof}

\subsection{Duality for null cuts}

\begin{theorem}\label{thm:completeness}
	The standard subspaces $H(N_C)$ of continuous $C$ satisfy duality (see Figure \ref{fig:setup}): 
	\begin{align*}
		H(N_C)=H(N_C)''=H(N_C'').
	\end{align*}
\end{theorem} 
\begin{proof}
	The first equality follows from the fact that $H(N_C)$ is a standard subspace.

	Next we show $H(N_C) \subset H(N_C'')$.
	By any shift $ta$ in the $x_1$-direction (or $x_-$-direction), we have $N_C \subset N_C'' - ta$ for $t>0$.
	Indeed, if $x \in N_C, y\in N_C'$, it has $x_- = 0$, and $(x - y)^2 = -2(x_+ - y_+)y_- - (\pmb{x}_\perp - \pmb{y}_\perp)^2 < 0$.
	Furthermore, as $x_+$ can be arbitrarily large, it must hold that $x_+ - y_+>0, y_->0$
	and hence adding $ta$ to $y$ just decreases the first term.
	Therefore, for any $y \in N_C'$, $y+ta$ is spacelike from $x \in N_C$, that is,
	$N_C \subset N_C'' + ta$.
	Therefore, by Lemma \ref{lm:isotony-null-space-region}, we have $H(N_C)\subset H(N_C'')$.
	
	Lastly we prove $H(N_C'') \subset H(N_C)''$.
	It is easy to see that $N_C^\dagger - ta \subset N_C'$
	for any $t>0$, where $a$ is a positive vector in the $x_1$-direction,
	hence we have $H(N_C^\dagger) \subset H(N_C')$ by Lemma \ref{lm:isotony-null-space-region}.
	As we have $H(N_C'') \subset H(N_C')'$ by locality and
	we have $H(N_C'') \subset H(N_C')' \subset H(N_C^\dagger)' = H(N_C)''$,
	where the last equality follows from Lemma \ref{lem:ABC} and \eqref{decomposition-subspace-complement-null-cut}.

\end{proof}

As an example, it is immediate to see that the causal completion of the zero null-cut $N_0 = N_{C_0}$ with $C_0(\tx_\perp) = 0$ is the right wedge: $N_0''=W_1$, hence $H(N_0)=H(N_0'')=H(W_1)$.

While it is possible to define $H(N_C)$ for non continuous curves (see Section \ref{concluding}),
the continuity condition Theorem \ref{thm:completeness} cannot be removed.
Indeed, if $C=\chi_{\mathbb{Q}^{D-1}}$ is the characteristic function of $\mathbb Q^{D-1}$ then $H(N_C)=H(W_R)$
since the direct integrals 
\begin{align*}
	H(N_C)=\int_{\RR^{D-1}}^{\oplus_\RR} H_\uone((\chi_{\mathbb{Q}^{D-1}}(\tx_\perp),+\infty))d^{D-1}\tx_\perp=\int_{\RR^{D-1}}^{\oplus_\RR} H_\uone((0,+\infty))d^{D-1}\tx_\perp=H(W_1)
\end{align*}
are equal up to a null measure set,  but  $H(N_C'')$ is strictly contained in $H(W_R)$ when, again, is considered the definition \eqref{local-subspace}.

Let us define the deformed wedge for a continuous function $C:\RR^{D-1}\rightarrow \RR$ by 
	\begin{align*}
		W_C:=\{\tx \in \RR^{D+1},x_-<0,x_+>C(\tx_\perp)\}.
	\end{align*} 
Then we have:
	\begin{proposition}\label{prop:N_C=W_C}
		$H(N_C)=H(W_C)$.
	\end{proposition}
	\begin{proof}
		It is clear from the definition that $N_C\subset W_C - tx_-$ for every $t>0$. From \autoref{lm:isotony-null-space-region}, it follows $H(N_C)\subset H(W_C)$.
		
		For the opposite direction, we will deduce $H(W_C) \subset H(N_C'')$. The claim follows from isotony and the previous theorem.
		Let $y\in W_C$, then we can write $y = x + a$ with $x \in N_C$ and $a=(0,a_-,0)$
		It is enough to prove $N_C-t a \subset N_C''$ for $t > 0$, or equivalently, $N_C\subset N_C''+t a$ for every $\e>0$.
		This is contained in the proof of \autoref{thm:completeness}.

	\end{proof}

\section{Modular operator for null cuts}\label{section-constant-null-cut}
 The following is of an independent interest, because it is an assumption of the limited version of QNEC \cite{CF18}.
 \begin{proposition}\label{pr:hsmi}
 Let $C_1, C_2:\RR^{D-1}\rightarrow \RR$ be continuous functions such that $C_1(\pmb{x}_\perp) < C_2(\pmb{x}_\perp)$.
Then the inclusion $H(N_{C_2})\subset H(N_{C_1})$ is a HSMI.   
\end{proposition}
\begin{proof}
Note that, if $K \subset H$ is a HSMI and $U$ is a unitary, then also $UK \subset UH$ is a HSMI by Lemma \ref{lem:cov}.

Recall that\footnote{With this slightly unfortunate notations, $W_1$ is the standard wedge in the $x_1$-direction
(and not the shifted wedge).} $H(W_+) \subset H(W_1)$ is a HSMI, where $C_+(\pmb{x}_\perp) = 1$ and $W_+ = W_{C_+}$ is the shifted wedge.
Put also $C_0(\pmb{x}_\perp) = 0$, so that $W_1 = W_{C_0}$.
By fibre-wise covariance (Proposition \ref{pr:null-covariance}), the inclusion
$H(N_{C_2})\subset H(N_{C_1})$ is unitarily equivalent to $H(W_+) = H(C_+)\subset H(C_0) = H(W_1)$.
Indeed, we can take $C_{\mathrm{dl}}(\pmb{x}_\perp) = \ln(C_2(\pmb{x}_\perp) - C_1(\pmb{x}_\perp))$
and $C_{\mathrm{tr}}(\pmb{x}_\perp) = C_1(\pmb{x}_\perp)$.
Then we have $T_{C_{\mathrm{tr}}}D_{C_{\mathrm{dl}}}H(W_1) = H(C_{\mathrm{tr}} + e^{C_{\mathrm{dl}}}C_0) = H(N_{C_1})$
and $T_{C_{\mathrm{tr}}}D_{C_{\mathrm{dl}}}H(W_+) = H(C_{\mathrm{tr}} + e^{C_{\mathrm{dl}}}C_+) = H(N_{C_2})$.
Then the desired HSMI follows from this unitary equivalence.

\end{proof}

In the rest of this Section, we will study the modular operator of the standard spaces $H(N_C)$.
We call $Q_S$ the unitary that results from the concatenation of $V_{\mathrm S}$ (cf.\! Section \ref{section-decomposition-one-particle-space}) and the fibre-wise unitary between $L^2(\RR,d\theta')$ and the $U(1)$-current mentioned in Section \ref{sec:U(1)-current}.

  \begin{proposition}\label{pr:modular-operator-zero-space}
   The unitary operator $Q_S$ between $\hm$ and the direct integral of $U(1)$-currents decomposes the modular operator of $C(N_0)$:
    \begin{align*}
     \log(\Delta_{H(N_0)})&=\log(\Delta_{H(W_R)})\\&=\Ad_{Q_S}\int^{\oplus}_{\RR^{D-1}} \log(\Delta_{H_{U(1)}(\RR_+)})d^{D-1}\pmb{x}_\perp.
    \end{align*}
  \end{proposition}
  \begin{proof}
   The causal complement of $N_0$ is the left wedge $W_L$:
    \begin{align*}
     N_0'
     &=\{x\in \RR^{D+1}; x_+< 0,x_->0,\pmb{x}_\perp\in \RR^{D-1}\}=W_L.
    \end{align*}
   By Theorem \ref{thm:completeness}, we know $H(N_0)=H(N_0'')=H(W_R)$.
   By the Bisognano-Wichmann property (HK\ref{bw}) in Section \ref{sec:oneparticle},
   we know that the modular group of the right wedge coincides with the representation of the boost group along $x_1$:
   		\begin{align*}
   			 \Delta_{H(W_R)}^{it}= U_m(\Lambda_{1}(-2\pi t)).
   		\end{align*}
   	By \eqref{eq:non-const-dilations} the representation of the boost group is equivalent to
   	the constant dilation group of the $U(1)$-current on each fibre of the spatial decomposition of $\hm$ and hence:
   		\begin{align*}
   			\Ad_{Q_S}\Delta_{H(W_R)}^{-it}=\int^{\oplus}_{\RR^{D-1}} U_\uone(2\pi t,0)\,d^{D-1}\pmb{x}_\perp.
   		\end{align*}
   	The statement follows from \autoref{prop:functional-calculus-decomposable} and the Bisognano-Wichmann property for the $U(1)$-current
   	(cf.\! Section \ref{sec:U(1)-current}).
  \end{proof}

The standard subspace of the constant null cut $C_s(\tx_\perp)=s$ is
  \begin{align*}
   H(N_{s})&=T_{C_s}H(N_0)
   =U_m(1, \mathfrak{t}_{x_+}(s))H(W_1)
   =H(\mathfrak{t}_{x_+}(s)W_1).
  \end{align*}
 That means that $H(N_s)$ equals the standard subspace associated with the wedge $\mathfrak{t}_{x_+}(s)W_1$. According to the Bisognano-Wichmann property, the associated modular group is the boost group leaving the wedge $\mathfrak{t}_{x_+}(s)W_1$ invariant:
  \begin{align*}
   \Delta_{H(N_s)}^{it}&= U_m(1,\mathfrak{t}_{x_+}(s))\Delta_{H(N_0)}^{it}U_m(1,\mathfrak{t}_{x_+}(s))^*\\
   &=U_m(1,\mathfrak{t}_{x_+}((1-e^{2\pi t})s))\Delta_{H(N_0)}^{it}.
  \end{align*}
  From this, we have the following well-known (cf.\! \autoref{lem:cov}) relation between generators.
    \begin{align}\label{generator-constant-null-cut}
     \log{\Delta_{H(N_s)}}=\log{\Delta_{H(N_0)}}+2\pi sP,
    \end{align}
   where $P$ is the generator of translations $U_m(\mathfrak{t}_{x_+}(s))$ along the light ray $x_+$.
 
 The formula \eqref{generator-constant-null-cut} has a natural generalization to arbitrary null-cuts.
 The key structures that we will utilize are half-sided modular inclusion of certain null-cut subspaces. In the following, we will use the notation for half-sided modular inclusions introduced in Section \ref{sec:ss}.

Now we state one of our main results.
\begin{theorem}\label{main-thm}
 Let $C:\RR^{D-1}\rightarrow\RR$ be a continuous function, then the generator of the modular group of the associated null-cut is
    decomposed as follows, where the equivalence is given by the operator $Q_{\mathrm S}$ of Proposition \ref{pr:modular-operator-zero-space}.
    \begin{align}\label{null-cut-modular-spatial-decom}
     \log(\Delta_{H(N_C)})\simeq \int^\oplus_{\RR^{D-1}}\left( \log(\Delta_{H_{U(1)}}(\RR_+))+ 2\pi C(\pmb{x}_\perp) P_{\pmb{x}_\perp}\right)\,d\pmb{x}_\perp,
    \end{align}
   where $P_{\pmb{x}_\perp}$ is the generator of translations on each fibre.
\end{theorem}

\begin{figure}[t]
	\centering
	\includegraphics[width=0.6\textwidth]{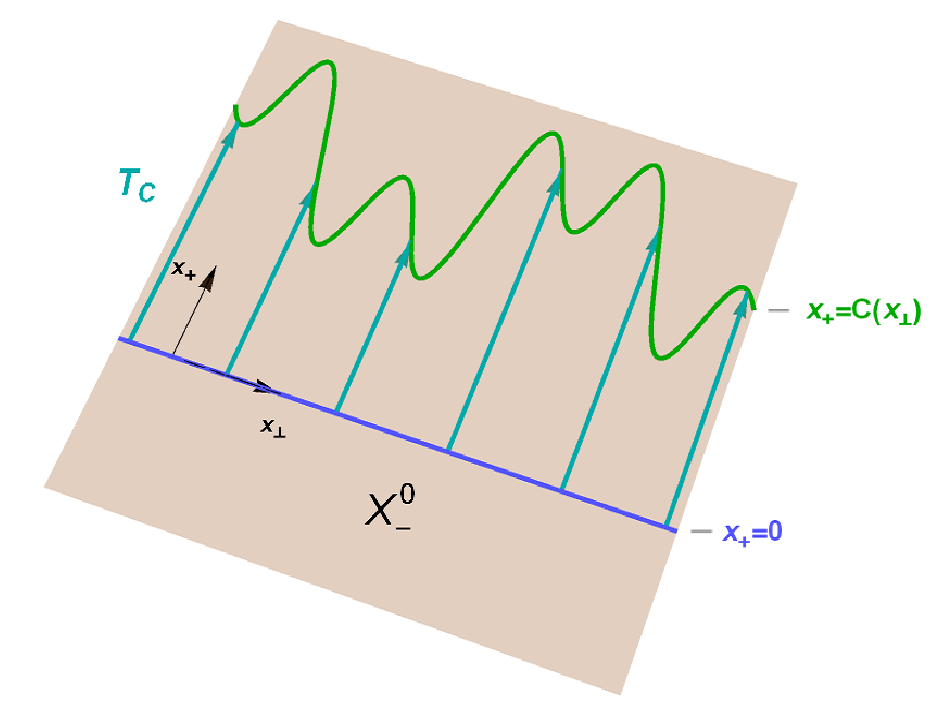}
	\caption{Geometric action of distorted lightlike translations.}
	\label{fig:modular-action}
\end{figure}
\begin{proof}
As we have $N_C = T_C N_0$ (with the notations of Lemma \ref{lm:isotony-null-space-region}),
we have $\Delta_{H(N_C)}^{it} = \Ad T_C (\Delta_{H(N_0)}^{it})$. Note that $T_C$ does the fibre-wise translation,
therefore by the decomposition through $Q_\mathrm{S}$,
\begin{align*}
 \Delta_{H(N_C)}^{it} &\simeq \int^\oplus_{\RR^{D-1}}
                                 \Ad U_{\uone}(0,\mathfrak{t}(C(\pmb{x}_\perp))(\Delta_{H_{U(1)}}(\RR_+)^{it}) \,d\pmb{x}_\perp\\
 &=\int^\oplus_{\RR^{D-1}}\exp\left(it\left(\log(\Delta_{H_{U(1)}}(\RR_+))+ 2\pi C(\pmb{x}_\perp) P_{\pmb{x}_\perp}\right)\right)\,d\pmb{x}_\perp,
\end{align*}
where the second equation is the fibre-wise transformation law of a simple HSMI (Theorem \ref{theorem-1-standard}).
By Proposition \ref{prop:functional-calculus-decomposable} and the Stone theorem,
this relation passes to the generators.

\end{proof}

Let us comment on the second quantized net $(\A, \mathrm{\Gamma}_+(U_m), \mathcal{F}_+(\hm))$.
By Proposition \ref{prop:secquant}, a HSMI of standard subspaces promotes to a HSMI of von Neumann algebras,
hence the following is an immediate consequence of Proposition \ref{pr:hsmi} and Theorem \ref{thm:completeness}.
\begin{corollary}
 Let $C_1, C_2:\RR^{D-1}\rightarrow \RR$ be continuous functions such that $C_1(\pmb{x}_\perp) < C_2(\pmb{x}_\perp)$.
Then the inclusion $\A(N_{C_2}'')\subset \A(N_{C_1}'')$ is a HSMI.   
\end{corollary}

It is not easy to write the second-quantized version of \eqref{null-cut-modular-spatial-decom}:
while the second quantization of the left-hand side is straightforward
$\mathrm{d\Gamma}_+(\log \Delta_{H(N_C)}) = \log \Delta_{\A(N_C''),\Omega}$,
where $\mathrm{d\Gamma}_+(A) = 0\oplus A \oplus (A\otimes \1 \oplus \1\otimes A)\cdots $ is the additive second quantization,
the direct-integral structure on the right-hand side translates to the \textit{continuous tensor product} structure
in the second quantization \cite{AW66, Napiorkowski71}.
Conceptually, it should be an ``integral'' of the second-quantized generators on fibres,
but we do not know how to formulate such an integral corresponding to \eqref{eq:set}.
Instead, we were able to formulate the HSMI without reference to the (undefined) stress-energy tensor
smeared on the null plane.

On the other hand, for general interacting models, we do not have this simple second-quantization structure,
and the validity of HSMI for the inclusions of null cut regions remains open, see Section \ref{concluding}.

\section{Relative entropy and Energy bounds}\label{entropy}

\subsection{Decomposition of the relative entropy for coherent states}\label{section-Entropy-NEC}

We recall that the formula for Araki relative entropy of a von Neumann algebra $\A\subset\mathcal B(\K)$ with respect to two vector states $\omega_1$ and $\omega_2$ is given by
$$S(\omega_1||\omega_2)=-\langle\xi,\log\Delta_{\eta,\xi} \xi\rangle$$
where $\xi,\eta\in\K$ implement $\omega_1$ and $\omega_2$ on $\A$, respectively and $S_{\eta,\xi} = J_{\eta,\xi}\Delta_{\eta,\xi}^\frac12$ is the relative Tomita opereator:
$$S_{\eta,\xi}:\A\xi\ni a\xi\mapsto a^*\eta\in\A\eta.$$
It is immediate to see that if $\eta=\xi$, then $\Delta_{\xi,\xi}$ is the modular operator of $\A$ with respect to $\xi$ and $S(\omega_1||\omega_1)=0$.

Let $H \subset \H$ be a standard subspace and $\mathrm{w(\psi)}$ be the Weyl operator of $\psi\in\H$ on $\Gamma_+(\H)$.
A state on a von Neumann algebra $\A \subset \B(\Gamma_+(\H))$ is called \textbf{coherent}
if it is given by $\omega_\psi(\cdot):=\omega(\mathrm w(\psi)\cdot \mathrm w(\psi)^*)$,
where $\omega$ is the Fock vacuum.

Following \cite{CLR20}, given a standard subspace $H\subset\H$, we can consider the following quantity
\begin{align*}
S_H(\psi)=\Im\langle \psi,P_Hi\log\Delta_H \psi\rangle,
\end{align*}
where $P_H$ is an unbounded real projection called \textit{cutting projection}. The cutting projection is determined by the modular operator and modular conjugation of $H$. Its explicit form is:$$P_H=a(\Delta_H)+ J_Hb(\Delta_H)$$
where $a(\lambda)=\lambda^{-1/2}(\lambda^{-1/2}-\lambda^{1/2})^{-1}$ and $b=(\lambda^{-1/2}-\lambda^{1/2})^{-1}$.\\
If $h\in H$, then 
	\begin{align*}
		S_H(\psi)=-\langle \psi,\log\Delta_H \psi\rangle.
	\end{align*}
We have the following relation with the relative entropy of second quantization algebra:
\[
 S_{R(H)}(\omega_\psi ||\omega)=S_H(\psi).
\]
where $S_{R(H)}(\omega_\psi||\omega)$ is the relative entropy of $\omega_\psi$ and $\omega$ with respect to $R(H)$
\cite[Proposition 4.2]{CLR20}.
\begin{lemma}\label{lem:relative-entropy-decomposition}
	Consider a standard subspace $H$ that is a direct integral of standard subspaces $H=\int_X^{\oplus_\RR} H(x)d\nu(\l)\subset \int_X^\oplus \K_\l d\nu(\l)$ with decomposable modular generator $\log(\Delta_H)=\int^\oplus_X \log(\Delta_{H(\l)})d\nu(\l)$. Moreover, let $\psi=\int^\oplus_X \psi(\l)d\nu(\l)\in \K$. Then the relative entropy of $\omega_\psi$ and the vacuum with respect to $H$ decomposes: 
		\begin{align*}
			S_{R(H)}(\omega_\psi||\omega)&=\int_X S_{R(H(\l))}(\omega_{_\Psi(\l)}(\l)||\omega(\l))d\nu(\l),
		\end{align*}
	where $\omega(\l)$ denotes the vacuum on the Fock space of $\K_\l$.
\end{lemma}
\begin{proof}
	Following the previous arguments, we have $S_{R(H)}(\omega_\psi||\omega)=S_H(\psi)$. By assumption the modular generator decomposes. It is straightforward to check that $J_H=\int^\oplus_X J_{H(\l)}d\nu(\l)$ decomposes as well. Applying \autoref{prop:functional-calculus-decomposable}, the cutting projection decomposes accordingly $P_H=\int^\oplus_X P_{H(\l)}\mu(\l)$. Thus
		$$\langle \psi,P_Hi\log\Delta_H \psi\rangle=\int_X \braket{\psi(\l),P_{H(\l)}i\log(\Delta_{H(\l)})\psi(\l)} d\nu(\l),$$ in particular
		\begin{align*}
			S_H(\psi)&=\im\int_X \braket{\psi(\l),P_{H(\l)}i\log(\Delta_{H(\l)})\psi(\l)} d\nu(\l)\\&=\int_X \im\braket{\psi(\l),P_{H(\l)}i\log(\Delta_{H(\l)})\psi(\l)} d\nu(\l)=\int_X S_{H(\l)}(\psi(\l)).
		\end{align*}
Since $S_{H(\l)}(\psi(\l))=S_{R(H(\l))}(\omega_{\psi(\l)}(\l)||\omega(\l))$, we conclude the argument.
\end{proof}
\begin{remark}
	The modular operator of the direct integral of standard subspaces always decomposes into the direct integral of modular operators of the fibre subspaces. This can be seen from the KMS condition for standard subspaces (see \cite[Proposition 2.1.8]{Longo08}) and the results in \autoref{appendix}.
\end{remark}

\subsection{Relative entropy for coherent states on the null plane}

For a single $U(1)$-current net, an explicit formula for the relative entropy was computed in \cite{LongoLocalised},
where an equivalent definition of the $U(1)$-current mentioned in Section \ref{sec:U(1)-current} is used.
We translate the results to the definition of the $U(1)$-current that we are working with. \\
Let $I_t=(t,+\infty)\subset\RR$, $H_\uone(I_t)$ be its standard subspace  and $\A(I_t)=R(H_\uone(I_t))$ be the second quantization von Neumann algebra. 
Consider the map $\beta_k(\mathrm w(\xi))=e^{-i\int_\RR k(x)L(x)\,dx}\mathrm w(\xi)$ with  $k,l\in C_0^\infty(\RR,\RR)$, $\hat{l}(0)=0$, $C_0^\infty(\RR,\RR)\ni L(x)=\int_{-\infty}^xl(s)ds$ the primitive of $l$ and $\xi\in H_\uone(I_t)$ being the one-particle vector of $l$ . Then $\b_k$ extends to an automorphism of the local algebra $\A(I_t)$, which is called a Buchholz-Mack-Todorov (BMT) automorphism \cite{BMT88}.
We denote this extension by $\b_k$ as well. One has \cite[Theorem 4.7]{LongoLocalised}:
\begin{align*}
S_{\A(I_t)}(\omega\circ\beta_k||\omega)=\pi\int_t^{+\infty}(x-t) k^2(x)dx,
\end{align*}
where $\omega$ is the vacuum state on the $U(1)$-current.

We can generalize BMT automorphisms to the direct integral of $U(1)$-currents.
Let $h,g\in C_0^\infty(X_-^0,\RR)$ with $\hat{g}(0)=0$ and $\xi=V_SE\delta(x_-)g$, we define the automorphism
\begin{align*}
\beta_h(\mathrm w(\xi))=e^{-i\int_{X_-^0} h(x)G(x)\,dx}\mathrm w(\xi). 
\end{align*}
Let us consider the relative entropy between $\omega$ and $\omega\circ \beta_h$ with respect to local algebras of null cuts:	
	\begin{proposition}\label{prop:rel-entr-coherent-state}
    We have
			\begin{align*}
				S_{R(H(N_C))}(\omega \circ \beta_h||\omega)=\pi \int_{\RR^{D-1}}\int_{C(\tx_\perp)}^\infty (x_+-C(\tx_\perp)) h(x_+,\tx_\perp)^2 dx_+ d\tx_\perp.
			\end{align*}
	\end{proposition}
	\begin{proof}
		
    As the support of $h$ is compact, $C$ is bounded on the restriction to $\tx_\perp$ of the support of $h$,
    and by adding an appropriate smooth function supported in $N_C^\dagger$, we may assume that
    $\int h(x_+,\tx_\perp)dx_+=0$ for every $\tx_\perp$ without changing the state $\omega\circ \beta_h$ on $R(H(N_C))$. Then $\delta(x_-)h$ is a thin test function and we call $\mathfrak{h}:=V_SE\delta(x_-)h$ its one-particle vector in the spatial decomposition.
    
	In this case, the BMT automorphism $\beta_h$ is generated by the adjoint action of the Weyl operator $\mathrm w(\mathfrak{h})$. The symplectic form can be computed fibrewise using \eqref{eq:U(1)-symplectic-form}
			\begin{align}
				\Ad \mathrm w(\mathfrak{h}) \left(\mathrm w(\xi)\right)&=e^{2i \im \braket{\mathfrak{h},\xi}}\mathrm w(\xi).\label{eq:BMT auto}\\
				\Im \braket{\mathfrak{h},\xi}= \int_{\RR^{D-1}} \Im\braket{\mathfrak{h}(\tx_\perp),\xi(\tx_\perp)}_{U(1)}d\tx_\perp&=\frac{1}{2}\int_{\RR^{D}} h(x_+,\tx_\perp)G(x_+,\tx_\perp)dx_+d\tx_\perp\label{eq:symplectic-form}.
			\end{align}
		Thus $\omega \circ \b_h$ is the coherent state $\omega_{\mathfrak{h}}$ and we can apply \autoref{lem:relative-entropy-decomposition}:
			\begin{align*}
				S_{R(H(N_C))}(\omega_{\mathfrak{h}}||\omega)=\int_{\RR^{D-1}}S_{\A(I_{C(\tx_\perp)})}(\omega_{\mathfrak{h}(\tx_\perp)}(\tx_\perp)||\omega_{\tx_\perp})d\tx_\perp.
			\end{align*}
		The coherent state $\omega_{\mathfrak{h}(\tx_\perp)}(\tx_\perp)$ on the $U(1)$ current is generated by the adjoint action of $\mathrm w_{U(1)}(\mathfrak{h}(\tx_\perp))$ on the fibrewise vacuum $\omega(\tx_\perp)$. This in turn coincides with the action of the BMT automorphism $\b_{h(\tx_\perp)}$ on the $U(1)$ vacuum $\omega({\tx_\perp})$.
		Therefore, we can insert the results for a single $U(1)$-current and have:
			\begin{align*}
					S_{R(H(N_C))}(\omega \circ 		\beta_h||\omega)&=\int_{\RR^{D-1}}S_{\A(I_{C(\tx_\perp)})}(\omega({\tx_\perp}) \circ \beta_{h(\tx_\perp)}||\omega({\tx_\perp}))d\tx_\perp\\
					&=\pi\int_{\RR^{D-1}}\int_{C(\tx_\perp)}^\infty 		(x_+-C(\tx_\perp)) h(x_+,\tx_\perp)^2 dx_+ d\tx_\perp.
			\end{align*}
	\end{proof}

\subsection{The ANEC and the QNEC}\label{sec:QNEC-ANEC}

\paragraph{Energy conditions in QFT}
One of the motivations to study quantum energy conditions comes from energy constraints in General Relativity, see e.g.~\cite{Few12}. Indeed, pointlike positivity of the stress energy tensor, assumed in general relativity to avoid singular spacetimes, is not possible in QFT.
It is not difficult to see that the stress energy smeared by a positive test function can not be positive but only bounded from below, see e.g.\! \cite{FH05}.
On the other hand, the ANEC claims that the the integral of the expectation value of the energy-momentum tensor in any physical state, along any complete, lightlike geodesic $\gamma$ is always non-negative (see e.g.\! \cite[(3)]{FLPW16},\cite{Verch00}):
\begin{align}\label{def:ANEC}
	\braket{\Phi|\int_\RR ds \,T_{\mu,\nu}(\gamma(s))k^\mu k^\nu \Phi }\geq 0.
\end{align}where  $\k^\mu$ is  the tangent of $\gamma.$ We can write $T_{++}$ when we consider the null $x_+$-direction $(x_+,0,\pmb{0})$.
It is expected that  \eqref{def:ANEC} is satisfied for a dense set of vectors in a certain formal sense.

The Quantum Null Energy (QNEC) is a local bound on the expectation value of the null energy density given by the von Neumann entropy of
the null cut (the following is only symbolic and we do not attempt to justify the expressions):
	\begin{align*}
		\braket{T_{++}(x)}_\Phi\geq \frac{1}{2\pi}S''_{\A(N_C),\Phi},
	\end{align*}
where $N_C$ is a null cut and $x$ is on the boundary of the cut $C$, and the derivative is taken
in the sense of the deformation of $C$ at the point $x$ (see \cite{BFKW19}).

In \cite{BFLW16} for null cuts $N_C$ and in \cite{LLS18} for deformed wedges $W_C$ it is claimed
that the QNEC is equivalent to the positivity of the second derivative of the relative entropy between the vacuum and
another state with respect to null-deformations:
\begin{align}\label{def:QNEC2}
	S''\geq 0.
\end{align}
The equivalence is reached through the expression of the modular operator as in  \eqref{eq:set} .
 With this motivation the positivity of $S''$ has been checked rigorously in some families of AQFT models (see e.g.~\cite{LongoLocalised,CLR20,Longocoherent, Pan19})
 and is expected to be a hold in general in the AQFT context.

In \cite{CF18}, the authors assume that the inclusions of null cut algebras are HSMI. For continuous functions $C_1\geq C_2$ and $A\geq 0$,
they showed
\begin{align}\label{def:QNEC3}
	\left.\partial_t^+ S_{R(N_{C_1+tA})}(\Psi||\Omega)\right|_{t=0}-\left.\partial_t^- S_{R(N_{C_2+tA})}(\Psi||\Omega)\right|_{t=0}\geq 0,
\end{align}
where $\partial_t^\pm$ denote half-sided partial derivatives.
This statement was proved for the class of vectors satisfying \eqref{def:ANEC} with respect to $T_{++}$ and having finite relative entropy with respect to the vacuum and the algebra of the zero null cut.

We will verify the QNEC \eqref{def:QNEC2} for the free scalar field and for some states explicitly.
In order to do this, we take the decomposition of the modular operator (cf. \autoref{main-thm}) and exploit the result by \cite{LongoLocalised} on each fiber.
In this sense, our findings constitute a generalization of the results to a continuum family of  $U(1)$-currents.

\paragraph{The ANEC}

To make contact with the physical literature,
let us consider distorted lightlike translations $T_A$ by $A\geq 0$, which generate half-sided modular inclusions of standard subspaces $H(N_{C+A})\subset H(N_C)$ of null-cuts.
It is claimed, e.g.\! \cite[(5.6) and (5.18)]{CTT17}, that the following operators are positive.
\begin{align*}
	P_{\tx_\perp}&=\int dx_+ T_{++}(x_+,\tx_\perp) \\
	H_A&=\int d^{D-1}\tx_\perp A(\tx_\perp)\int dx_+ T_{++}(x_+,\tx_\perp) = \int d^{D-1}\tx_\perp A(\tx_\perp)P_{\tx_\perp}.
\end{align*} 
Furthermore, it is argued that positivity of $H_A$ should imply ANEC \eqref{def:ANEC} in \cite[Section 3.2]{FLPW16}.

Note that the last expression does not involve $T_{++}$, and we can make sense of it in the free field if we interpret these relations at the one-particle level.
Indeed, it is simply a weighted integral of the generator $P_{\tx_\perp}$ of translations in the $\uone$-current,
which is positive. The operator $H_A$ is also positive as it is the second quantization of this positive operator.

Following \cite{LongoLocalised}, we define the vacuum energy associated to null deformations by $A \geq 0$ of a normal and faithful state $\varphi$,
that has a vector representative $\eta$, by
\begin{align}
	E_A(\eta)=\braket{\eta,H_{A}\eta}\label{ANEC},
\end{align}
where $H_A=\frac{1}{2\pi}\big(\log(\Delta_{R(H(N_{C+A})),\Omega})$ $-\log(\Delta_{R(H(N_C)),\Omega})\big)$.
The expression \eqref{ANEC} is positive for the dense set of vectors in $\D(H_A)$ (cf. \autoref{appendix} for explicit form of $\D(H_A)$) since $H_A$ is a positive operator. In \cite{LongoLocalised} and \cite{CF18}, the positivity of \eqref{ANEC} is considered as form of the ANEC.

\paragraph{The QNEC}
A BMT automorphism of the form \eqref{eq:BMT auto} generates a coherent state $\omega_h(\cdot)=\omega(\mathrm w(\mathfrak h)\cdot \mathrm w(\mathfrak h)^*)$
when restricted to a null cut. The GNS representation space of $\omega\circ \beta_h$ is the Fock space. The representation of $\tP$ transforms by the adjoint action of $\mathrm w(\mathfrak h)$. We call this representation the $\beta_h$-representation.

Let $\omega\circ \beta_{h_1}$ and $\omega\circ \beta_{h_2}$ be two coherent states that are unitarily generated by the adjoint action of $\mathrm w(\mathfrak{h}_1)$ and $\mathrm w(\mathfrak{h}_2)$, respectively. The relative entropy is between these states is:
\begin{align*}
	S(\omega\circ \beta_{h_1}||\omega\circ \beta_{h_2})&=S(\omega\circ \beta_{h_1}||\omega \cdot \Ad \mathrm w(\mathfrak{h}_2))=S(\omega\cdot \Ad \mathrm w(\mathfrak{h}_1-\mathfrak{h}_2)||\omega)=S(\omega\circ \beta_{h_1-h_2}||\omega).
\end{align*}
To study the relative entropy between these states, we can therefore restrict to $h_2=0$, i.e. $\omega \circ \beta_{h_2}$ being the vacuum state.

A distorted lightlike translation $T_{tA}$ by $A$ maps $H(N_C)$ to $H(N_{C+tA})$ and the relative entropy of $\omega\circ \b_h$ and $\omega$ changes accordingly (cf. \autoref{prop:rel-entr-coherent-state}). The differentiation with respect to the deformation parameter $t$ gives:
\begin{align}
	\frac{d}{dt}S_{R(N_{C+tA})}(\omega\circ\beta_h||\omega)&=-\pi \int_{\RR^{D-1}}\int_{C(\tx_\perp)+tA(\tx_\perp)}^\infty A(\tx_\perp)h(x_+,\tx_\perp)^2 dx_+d\tx_\perp\label{eq:first-derivative-relative-entropy}\\ 
	\frac{d^2}{dt^2}S_{R(N_{C+tA})} (\omega\circ\beta_h||\omega)&=\pi \int_{\RR^{D-1}} A(\tx_\perp)^2 h(C(\tx_\perp)+tA(\tx_\perp),\tx_\perp)^2 d\tx_\perp\geq 0
	\nonumber
\end{align}
We applied the differentiation fibrewise and used 
\begin{align}
	\frac{d}{dt}\int_{C+tA(\tx_\perp)}^{\infty}(x_+-(C+tA(\tx_\perp)) h(x_+,\tx_\perp)^2 dx_+ &=\int_{C+tA(\tx_\perp)}^\infty A(\tx_\perp)h(x_+,\tx_\perp)^2 dx_+\label{eq:first-fibre-t-derivative}\\
	\frac{d^2}{dt^2}\int_{C+tA(\tx_\perp)}^\infty A(\tx_\perp)h(x_+,\tx_\perp)^2 dx_+&= A(\tx_\perp)^2 h(C+tA(\tx_\perp),\tx_\perp)^2 \label{eq:second-fibre-t-derivative}.		
\end{align}
To justify the interchange of the $\tx_\perp$-integral and the $\frac{d}{dt}$-derivative, we verify that $\eqref{eq:first-fibre-t-derivative}, \eqref{eq:second-fibre-t-derivative}\in L^1(\RR^{D-1},d^{D-1}\tx_\perp)$. We estimate $\eqref{eq:first-fibre-t-derivative}\leq \norm{A(\tx_\perp)h(x_+,\tx_\perp)^2}_1$. Since $A$ is continuous and $h$ is compactly supported and smooth, the estimate is bounded and compactly supported in $\RR^{D-1}$ and as such it is integrable. For $\eqref{eq:second-fibre-t-derivative}$, we note that $C+tA(\tx_\perp)$ is continuous on $\RR^{D-1}$ and accordingly maps compact sets to compact sets. Hence, the product in $\eqref{eq:second-fibre-t-derivative}$ is compactly supported in $\RR^{D-1}$ and bounded and therefore integrable.  

In summary, we have proven the following form of the Quantum Null Energy Condition:
\begin{corollary}
	For the coherent states $\omega\circ\beta_h$ considered here, the QNEC holds:
		\begin{align*}
			\frac{1}{2\pi}\frac{d^2}{dt^2}S_{R(N_{C+tA})} (\omega\circ\beta_h||\omega)\geq 0.
		\end{align*}
	This inequality is not saturated at every point of positive energy density. Following the arguments in \autoref{thm:completeness} and \autoref{prop:N_C=W_C}, we can replace the region $N_{C+tA}$ with $N_{C+tA}''$ and $W_{C+tA}$, respectively.
\end{corollary}

As expected in \cite{BFKLW16}, we recover the ANEC by integrating the QNEC along a null-direction for coherent states by \eqref{ANEC}: 
	\begin{align*}
		&\frac{1}{2\pi}\int_\RR dt\frac{d^2}{dt^2}S_{R(N_{C+tA})} (\omega\circ\beta_h||\omega) \\
		&=\frac{1}{2}\int_\RR dt \int_{\RR^{D-1}}A(\tx_\perp)^2h(C(\tx_\perp)+tA(\tx_\perp),\tx_\perp)^2d\tx_\perp\\
		&=\frac{1}{2}\int_\RR dx_+ \int_{\RR^{D-1}}A(\tx_\perp)h(x_+,\tx_\perp)^2d\tx_\perp\\
	&=\im\braket{\mathfrak{h},iPA\mathfrak{h}} = i \left.\frac{d}{ds}\right|_{s=0}\braket{\mathfrak h,U_A(s)\mathfrak h} \\
	&=i\left.\frac{d}{ds}\right|_{s=0}\left(e^{-\frac{1}{2}(\norm{\mathfrak h}^2+\norm{U_A(s)\mathfrak h}^2)} e^{\braket{\mathfrak h,U_A(s)\mathfrak h}} \right)
	= i\left.\frac{d}{ds}\right|_{s=0}\braket{\mathrm w(\mathfrak h)\Omega,\Gamma(U_A(s))\mathrm w(\mathfrak h)\Omega } \\
		&=\braket{\mathrm w(\mathfrak{f})\Omega,H_A\mathrm w(\mathfrak{h})\Omega},
	\end{align*}
where we used that the self-adjoint fibrewise "momentum operator" $\int_{\RR^{D-1}}^\oplus A(\tx_\perp)P_{\tx_\perp}d\tx_\perp$ is the generator of $U_A(s)$,
similar arguments as in \eqref{eq:U(1)-symplectic-form} and \eqref{eq:symplectic-form} and the general identity \eqref{eq:second-quantization-identity} for the Weyl unitaries.
This is in agreement with \cite[Corollary 3.10 and (46)]{LongoLocalised}).

We interpret $A(\tx_\perp)h(x_+,\tx_\perp)^2$ as vacuum energy density at the point $(x_+,\tx_\perp) \subset X_-^0$ of the state $\mathrm w(\mathfrak h) \Omega$ with respect to null-deformations $A$ . This is the vacuum energy density of the state $\mathrm w(\mathfrak h(\tx_\perp))\Omega_{\tx_\perp}$ of the $U(1)$-current at $\tx_\perp$-fibre with respect to deformations by $A(\tx_\perp)$ (cf. \cite{LongoLocalised}).
With this interpretation, the energy averaged over the a null-cut $N_C$ is:
	\begin{align}
		E_{A,\b_h}( N_C)=\int_{\RR^{D-1}}\int_{C(\tx_\perp)}^\infty A(\tx_\perp)h(x_+,\tx_\perp)^2d\tx_\perp\label{eq:energy-null-cut},
	\end{align}
and we have the following connection between the first derivative of the relative entropy and the energy localised in the null-cut by comparing \eqref{eq:first-derivative-relative-entropy} and \eqref{eq:energy-null-cut}:
	\begin{align*}
		\frac{d}{dt}S_{R(N_{C+tA})}(\omega\circ\beta_h||\omega)=E_{A,\b_h}(N_{C+tA}).
	\end{align*}

\paragraph{Strong superadditivity of relative entropy}

Consider two null-cuts $N_{C_1}$ and $N_{C_2}$ and the null-cuts $N_{C_\cup}$ and $N_{C_\cap}$ generated by $C_\cup(\tx_\perp)=\min\{C_1(\tx_\perp),C_2(\tx_\perp)\}$ and $C_\cap(\tx_\perp)=\max\{C_1(\tx_\perp),C_2(\tx_\perp)\}$ respectively. 
Then the strong superadditivity of relative entropy is: 
	\begin{align*}
		S_{R(N_{C_\cup})}(\Psi||\Omega) + S_{R(N_{C_\cap})}(\Psi||\Omega) 	\geq S_{R(N_{C_1})}(\Psi||\Omega) + S_{R(N_{C_2})}(\Psi||\Omega).
	\end{align*}
As proven in \cite[Section 3.2]{CF18} the QNEC \eqref{def:QNEC3} implies this strong superadditivity
for a state $\Psi$ with finite QNEC \eqref{def:QNEC3}.
We can show that the states considered in Section \ref{sec:QNEC-ANEC} saturate the strong superadditivtiy of relative entropy.
Indeed, we apply \autoref{prop:rel-entr-coherent-state} and have:
	\begin{align*}
			S_{R(H(N_{C_\cup}))}&(\omega \circ \beta_h||\omega)+S_{R(H(N_{C_\cap}))}(\omega \circ \beta_h||\omega)-S_{R(H(N_{C_1}))}(\omega \circ \beta_h||\omega)-S_{R(H(N_{C_2}))}(\omega \circ \beta_h||\omega)\\
			&=\pi \int_{\RR^{D-1}}\Big(\int_{C_\cup(\tx_\perp)}^\infty (x_+-C_\cup(\tx_\perp)) +\int_{C_\cap(\tx_\perp)}^\infty (x_+-C_\cap(\tx_\perp))\\
			&\qquad-\int_{C_1(\tx_\perp)}^\infty (x_+-C_1(\tx_\perp))-\int_{C_2(\tx_\perp)}^\infty (x_+-C_2(\tx_\perp)) \Big)  h(x_+,\tx_\perp)^2 dx_+ d\tx_\perp=0
	\end{align*}
since $C_\cup(\tx_\perp)+C_\cap(\tx_\perp)-C_1(\tx_\perp)-C_2(\tx_\perp)=0$ for all $\tx_\perp$.

\section{Concluding remarks}\label{concluding}
Let us close this paper with a few more remarks.
\begin{itemize}
 \item
 It is possible to extend the definition of local subspaces on the null plane to measurable functions $C$. The strategy is to loosen the assumption on "thin test functions" (cf.\! Section \ref{sec:local-subspace-of-null-plane}) in the sense that $\tx_\perp \mapsto g(x_+,\tx_\perp) \in L^1(\RR^{D-1})\cap L^2(\RR^{D-1})$ for almost every $x_+$. The real subspaces are covariant with respect to distorted lightlike translations and dilations (cf.\! \eqref{eq:non-const-translation} and \eqref{eq:non-const-dilations}) for measurable functions $C$. Also in  this case $T_C$ and $D_C$ are determined by measurable vector fields of bounded operators since Remark \ref{rem:TCDC} obviously extends. It constitutes an extension of the present analysis in the sense that for continuous functions the definitions coincide.
Many results such as the decomposition of the subspaces (\autoref{decomposition-subspace-null-general}) and modular operators (\autoref{main-thm}) hold for the resulting subspaces as well. However, isotony (\autoref{lm:isotony-null-space-region}) and duality for null-cuts (\autoref{thm:completeness}) do not hold for arbitrary measurable functions.
\item 
The modular operator of a direct integral of standard subspaces decomposes into the direct integral of modular operators.
Therefore, the modular operator of $H(R)$ from \autoref{decomposition-subspace-null-general} decomposes into the direct integral of $U(1)$-modular operators of the interval $(C_1(\tx_\perp),C_2(\tx_\perp))$.
\item The fact that there are observables that can be restricted to null plane shows
that the minimal localization region in the sense of \cite{Kuckert00} can have empty interior.
\item In a general Haag-Kastler net, we do not expect that there are sufficiently many observables
(e.g.\! in the sense of the Reeh-Schlieder property) that can be restricted to the null plane.
With additional assumptions (including the $\CC$-number commutation relations),
it is shown that only free fields can be directly restricted on the null plane \cite{Driessler76-2},
cf.\! \cite[Section 5]{Wall12} \cite[A]{BCFM15}.
In addition, in the two-dimensional spacetime, there are interacting Haag-Kastler nets \cite{Tanimoto14-1}
where observables on the lightray generate a proper subspace of the Hilbert space from the vacuum
\cite[Section 5.3, Trivial examples]{BT15}.
In these cases, different ideas are required to justify \eqref{eq:set}.

On the other hand, it has been suggested that there could be bounded operators localized on the null plane
in the sense of Haag duality \cite{Schroer05}. If so, two-dimensional conformal field theory is hidden
on each lightlike fibre \cite{BLM11}.
\end{itemize}

\subsubsection*{Acknowledgements}
We thank Roberto Longo and Aron Wall for inspiring discussions.

VM is supported by Alexander-von-Humboldt Foundation and was supported  by the European Research Council Advanced Grant 669240 QUEST until March 2021. BW is supported by the INdAM Doctoral programme and has received funding from the European Union's 2020 research and innovation programme under the Marie Sklodowska-Curie grant agreement No 713485.

We acknowledge the MIUR Excellence Department Project awarded to
the Department of Mathematics, University of Rome ``Tor Vergata'' CUP E83C18000100006 and the University of
Rome ``Tor Vergata'' funding scheme ``Beyond Borders'' CUP E84I19002200005.

\appendix
\section{Decomposable Functional calculus}\label{appendix}

	Assume the same notations as in Section \ref{direct}. We call an (possibly unbounded) operator $T$ on the direct integral of Hilbert spaces $\K=\int^\oplus_X \K_\l d\nu(\l)$ decomposable when there is a function $\l \mapsto T_\l$, such that $T_\l$ is an (possibly unbounded) operator on $\K_\l$, and for each $\xi \in \D(T)$, $\xi(\l) \in \D(T_\l)$ and $(T\xi)(\l)=T_\l\xi(\l)$ for $\nu$-almost every $\l$. 
		
		\begin{theorem}\label{prop:functional-calculus-decomposable}
			Let $T$ be a self-adjoint operator on the direct integral of Hilbert spaces $\K=\int^\oplus_X K_\l d\nu(\l)$.
			The following are equivalent:	
				\begin{enumerate}
					\item\label{4} The projection-valued measure $E^T$ associated to $T$ decomposes in the direct integral of projection-valued measures:
						\begin{align}\label{eq:decomposed-spectral-measure}
							E^T=\int^\oplus_X E^{T_\l}d\nu(\l)=:E
						\end{align}
					for some self-adjoint operators $T_\l$ on $\K_\l$.
					\item\label{3} The Borel functional calculus of $T$ decomposes in the direct sum of Borel functional calculi of self-adjoint operators $T_\l$ on $\K_\l$:
						\begin{align*}
							f(T)=\int^\oplus_X f(T_\l)d\nu(\l).
						\end{align*}
					\item\label{1} $T=\int^\oplus_X T_\l d\nu(\l)$ is decomposable. Each $T_\l$ is a self-adjoint operator on $\K_\l$.
					\item\label{2} The one-parameter group of unitaries $U(t):=e^{itT}$ decomposes:
						\begin{align*}
							U(t)=\int^\oplus_X U_\l(t)d\nu(\l),
						\end{align*}
					where $U_\l(t)$ is a one-parameter group of unitaries on $\K_\l$.

				\end{enumerate}
		\end{theorem}	
		\begin{proof}
			\ref{3} $\Rightarrow$ \ref{1},\ref{2}: clear.
			
			\ref{1} $\Leftrightarrow$ \ref{4} $\Rightarrow$ \ref{3}:
			We show that $E$ (cf.\! \eqref{eq:decomposed-spectral-measure}) is the spectral measure associated to $\int^\oplus_X T_\l d\nu(\l)$. It is a projection-valued measure defined by:
				\begin{align*}
					\Omega \mapsto E(\Omega):=\int^\oplus_X E^{T_\l}(\Omega)d\nu(\l).
				\end{align*}
			That is because the multiplication of decomposable operators is defined fibre-wise. It is clear that each element is a projection and that $E(\emptyset)=0$, $E(\RR)=1$ and $E(\Omega_1 \cap \Omega_2)=E(\Omega_1)E(\Omega_2)$.
			For mutually disjoint $\Omega_n$ the associated projections are orthogonal. It follows that their sum is again a projection. Hence, we can use dominated convergence, to deduce the convergence $E(\cup^N \Omega_n)\xrightarrow{N\rightarrow \infty}E(\Omega)$ for a countable union $\Omega=\cup \Omega_n$ from the fibre-wise convergence. The latter is true because $E^{T_\l}$ is a spectral measure.
			The spectral measure $E_\varphi$ for $\varphi\in \K$ is for a Borel set $\Omega$:
				\begin{align*}
					E_\varphi(\Omega)=\braket{\varphi,E(\Omega)\varphi}
					&=\int_X E^{T_\l}_{\varphi_\l}(\Omega) d\nu(\l).
				\end{align*}
			For integration, we denote it by $dE_\varphi(x)$. For a simple
			function $s=\sum a_i \chi_{\Omega_i}$ on a compact set $\Omega=\cup_{i} \Omega_i$ it holds:
				\begin{align*}
					\int_\RR s(x) dE_\varphi(x)&=\sum\limits_{i}a_i\int_X E_{\varphi_\l}^{T_\l}(\Omega_i)d\nu(\l)\\
					&=\int_X \int_\RR  s(x) dE^{T_\l}_{\varphi_\l}(x) 	d\nu(\l)=\braket{\varphi,\left(\int^\oplus_X s(T_\l)d\nu(\l)\right) \varphi}.
				\end{align*}
			The integral of a general bounded Borel function is the difference of the integrals over the positive and negative part.
			The integral over a positive bounded Borel function $f$ with respect to $dE_\varphi$ is defined as the supremum of integral over simple functions $s$ that are locally bounded by $f$:
				\begin{align*}
					\int_\RR f(x) dE_\varphi(x)&=\sup\limits_{0\leq s\leq f}\int_X \int_\RR s(x) dE^{T_\l}_{\varphi_\l}(x) d\nu(\l)\leq \int_X \sup\limits_{0\leq s\leq f}\int_\RR s(x) dE^{T_\l}_{\varphi_\l}(x) d\nu(\l)
				\end{align*}
			For a positive, Borel measurable function $f$, there is a non-decreasing sequence $s_n$ of simple functions converging from below pointwise. We can utilize dominated convergence (spectral measures are finite) to show 	
			$$\lim\limits_{n\rightarrow \infty}\int_X \int_\RR s_n(x) dE_{\varphi_\l}(x) d\nu(\l)=\int_X \int_\RR f(x) dE_{\varphi_\l}(x) d\nu(\l).$$ 
			Hence, we have for bounded Borel functions:
				\begin{align}\label{eq:Riesz-Markov-functionals}
					\int_X \int_\RR f(x) dE^{T_\l}_{\varphi_\l}(x) d\nu(\l)
				=\int_\RR f(x) dE_\varphi(x).
				\end{align}
			For each $\varphi \in \K$, the measure $E_\varphi$ is inner regular and countably additive. The second property was already discussed. The inner regularity follows from the finiteness of $E_\varphi$. Hence, the two integrals in \eqref{eq:Riesz-Markov-functionals} define the same bounded, linear functional on the compactly supported, continuous functions. From the uniqueness in the Riesz-Markov theorem (cf.\! \autoref{lem:Riesz-Markov}), we deduce the equality of the measures. This means for each $\varphi \in \K$ and Borel function $h$:
				\begin{align*}
					\int_\RR h(x)dE_\varphi=\int_X\int_\RR h(x)dE_{\varphi_\l}^{T_\l}d\nu(\l)=\braket{\varphi, \left( \int^\oplus_X h(T_\l)d\nu(\l) \right)\varphi}.
				\end{align*}
			In the usual notation, this amounts to:
				\begin{align*}
					\int_\RR h(x)dE(x)=\int^\oplus_X h(T_\l) d\nu(\l).
				\end{align*}
			The domain of the operator $h(T)$ is:
				\begin{align*}
					\D(h(T))&=\{\varphi\in \K; \int_\RR |f(k)|^2 dE_{\varphi}^T(k)<\infty   \}\\
					&=\{\varphi\in \K; \int_X\int_\RR |f(k)|^2 dE_{\varphi(\l)}^{T_\l}(k)d\nu(\l)<\infty   \}\\
					&= \{\varphi\in \K; \int_X\norm{h(T_\l)\varphi(\l)}^2 d\nu(\l)<\infty   \}  
					=\int^\oplus_X \D(h(T_\l))d\nu(\l).
				\end{align*}
			For a closed operator $A$ on a Hilbert space $\H$ its domain $\D(A)$ is a Hilbert space   w.r.t. the Graph inner product:
				\begin{align*}
					\braket{\xi,\eta}_{A}:=\braket{\xi,\eta}_{\H}+\braket{A\xi,A\eta}_{\H}.
				\end{align*}
			Since $T_\l$ is self-adjoint for ($\nu$-almost) every $\l$, the domain of $h(T_\l)$ is a Hilbert space w.r.t. the Graph inner product for almost every $\l$.  The domain $\D(h(T))$ embeds in $\K$ trivially. 
			If $T$ is decomposable, then we have for all vectors $\varphi \in \K$:
				\begin{align*}
					\int_\RR \l dE_\varphi^T(\l)=\int_X\int_\RR \l dE_{\varphi_\l}^{T_\l}d\nu(\l).
				\end{align*}
			
			\ref{2} $\Rightarrow$ \ref{4}:
				Let $\xi, \eta \in \K$ and $g\in \mathscr{S}(\RR)$. Then by Fubini's theorem the operator $g(T)$ decomposes:
			\begin{align*}
				\braket{\xi, g(T)\eta}
				&=\int_\RR \tilde{g}(t)\braket{\xi,U(t)\eta} dt\\
				&=\int_X \braket{\xi(\l),\left(\int_\RR \tilde{g}(t)U_\l(t)dt  \right)\eta(\l)} d\nu(\l)\\
				&=\braket{\xi,\left(\int^{\oplus}_X g(T_\l)d\nu(\l)  \right)\eta}.
			\end{align*}
			Every characteristic function $\chi_S$ of a set $S$ with finite measure can be expressed as the pointwise limit of a sequence $g_n\subset \mathscr{S}(\RR)$ of uniformly bounded functions. In the functional calculus this amounts to the strong convergence:
			\begin{align*}
				g_n(T)\xrightarrow{n\rightarrow \infty}\chi_{S}(T)=E^T(S),
			\end{align*}
			where $E^T(S)$ denotes the spectral projection of $T$ with respect to the set $S$.

			Since $g_n$ is uniformly bounded, we can apply the theorem of dominated convergence:
			\begin{align*}
				\braket{\xi,E^T_{S}\eta}&=\lim\limits_{n\rightarrow \infty}\int_X\braket{\xi(\l),g_n(T_\l)\eta(\l)}d\nu(\l)\\
				&=\int_X\braket{\xi(\l),\chi_{S}(T_\l)\eta(\l)}d\nu(\l)\\
				&=\braket{\xi,\left(\int^\oplus_X E^{T_\l}(S)d\nu(\l)\right)\eta}.
			\end{align*}
			Hence, the spectral measure $E_{\varphi}^{T}$ associated to any vector $\varphi \in \K$ decomposes in the following sense for any bounded Borel function $f$:
			\begin{align}
				\braket{\varphi,f(T)\varphi}&=\int_\RR f(k) dE_{\varphi}^T(k)\label{Riesz-Markov-functional-1}\\&=\int_X\int_\RR f(k) dE_{\varphi(\l)}^{T_\l}(k)
				=\braket{\varphi,\left(\int^\oplus_X f(T_\l)d\nu(\l)\right) \varphi}\label{Riesz-Markov-functional-2}.
			\end{align}
			The two expression \eqref{Riesz-Markov-functional-1} and \eqref{Riesz-Markov-functional-2} define the same linear functional on the continuous, compactly supported functions. By the Riesz-Markov theorem such a functional is determined a unique countably additive, inner regular measure on $\RR$. Since both measure in \eqref{Riesz-Markov-functional-1} and \eqref{Riesz-Markov-functional-2} are countably additive and inner regular, they coincide.
		\end{proof}
	
		\begin{lemma}\label{lem:Riesz-Markov}[Riesz-Markov theorem]
			Let $X$ be a locally compact Hausdorff space and $\Psi$ a bounded linear functional on $\C_c(X)$, then there exists a unique inner regular and countably additive measure $\mu$ on $X$ such that:	
				\begin{align*}
					\Psi(f)=\int_X f d\mu.
				\end{align*}
		\end{lemma}

{\small
\newcommand{\etalchar}[1]{$^{#1}$}
\def\cprime{$'$} \def\polhk#1{\setbox0=\hbox{#1}{\ooalign{\hidewidth
  \lower1.5ex\hbox{`}\hidewidth\crcr\unhbox0}}} \def\cprime{$'$}
}

\end{document}